\documentclass[hidelinks,reprint,aps,superscriptaddress,nofootinbib,longbibliography]{revtex4-1}
\usepackage{newtxtext}
\usepackage{graphicx}
\usepackage{amsthm}   
\usepackage[bottom,flushmargin]{footmisc}
\usepackage{amsmath} 
\usepackage{amssymb}  
\usepackage{mathrsfs} 
\usepackage{stmaryrd} 
\usepackage{nicefrac}

\usepackage{enumerate}
\usepackage{appendix} 
\usepackage{enumitem}
\usepackage{hyperref}
\usepackage{comment}
\usepackage[capitalize]{cleveref}
\usepackage{mathtools}

\usepackage{color}  
\RequirePackage[dvipsnames,usenames]{xcolor}

\DeclareMathOperator{\supp}{supp}

\newtheorem{theorem}{Theorem}

\newtheorem{lemma}[theorem]{Lemma}

\newtheorem{proposition}[theorem]{Proposition}

\usepackage[all]{xy}
\usepackage{thmtools}
\usepackage{thm-restate}

\usepackage{cleveref}

\usepackage[T1]{fontenc}
\usepackage[latin9]{inputenc}
\usepackage{mathrsfs}
\usepackage{bm}
\usepackage{mathtools}
\usepackage{thm-restate}

\newcommand{\ktlntwo}{kT\ln2}
\newcommand{\kconst}{\kappa}

\newcommand{\bitsToEntropy}[1]{\ln 2 \cdot #1}

\newcommand{\bitsToHeat}[1]{\ktlntwo \cdot #1}
\newcommand{\bitsToHeatParen}[1]{\ktlntwo \, [#1]}
\newcommand{\eplog}{\ln}
\newcommand{\boltz}[1]{e^{-#1/kT}}
\def\logZ{\ln Z}
\newcommand{\entropyToHeat}[1]{kT\,#1}
\newcommand{\entropyToHeatParen}[1]{kT\,[#1]}
\newcommand{\heatToEntropy}[1]{#1/{kT}}
\newcommand{\heatToBits}[1]{#1/(\ktlntwo)}

\newcommand{\uY}{m_Y}
\newcommand{\BB}{{\{0,1\}}^*}

\newcommand{\R}{{\mathbb{R}}}

\newcommand{\printycmd}{\text{``}\texttt{print `}y\texttt{`}\text{''}}
\def\termdef{\emph}

\def\dom{\mathrm{dom\,}}
\def\img{\mathrm{img\,}}
\def\N{\mathbb{N}}

\def\ft{t_f}

\def\EP{\Sigma}
\def\prior{w}

\def\sX{\mathcal{X}}
\def\pinit{p_X}
\def\pfin{p_Y}

\def\px{\pinit(\inp)}
\def\pfx{p_{f(X)}(f(x))}

\def\pcoin{\pinit^\mathrm{coin}}
\def\pfincoin{\pfin^\mathrm{coin}}
\def\Qcoin{\Q_\mathrm{coin}}

\def\Qopt{\Q_\mathrm{dom}}

\newcommand{\out}{y}
\newcommand{\inp}{x}
\newcommand{\Q}{Q}
\newcommand{\G}{G}

\newcommand{\T}{M}
\newcommand{\U}{U}
\newcommand{\func}[1]{\phi_{#1}}
\newcommand{\funcT}{\phi_M}
\newcommand{\funcU}{\phi_U}

\newcommand{\pXgY}{p_{X\vert f(X)}}

\newcommand{\priorXgY}{\prior_{X\vert f(X)}}
\newcommand{\pY}{p_{Y}}

\begin{document}

\title{Thermodynamic costs of Turing Machines}

 \author{Artemy Kolchinsky}
  \affiliation{Santa Fe Institute, 1399 Hyde Park Road, Santa Fe, NM 87501, USA}

 \author{David H. Wolpert}
\altaffiliation{Complexity Science Hub, Vienna}
 \altaffiliation{Arizona State University}
  \affiliation{Santa Fe Institute, 1399 Hyde Park Road, Santa Fe, NM 87501, USA}
 \affiliation{{\tt http://davidwolpert.weebly.com}}

\begin{abstract}
Turing Machines (TMs) are the canonical model of computation in computer science and physics.  We combine techniques from algorithmic information theory and stochastic thermodynamics to analyze the thermodynamic costs of TMs.  
We consider two different ways of realizing a given TM with a physical process. The first realization is designed to be thermodynamically reversible when fed with random input bits. The second realization is designed to generate less heat, up to an additive constant, than any  realization that is computable (i.e., consistent with the physical Church-Turing thesis). 
We consider three different thermodynamic costs: the heat generated when the TM is run on each \emph{input} (which we refer to as the ``heat function''), the minimum heat generated when a TM is run with an input that results in some desired \textit{output} (which we refer to as the ``thermodynamic complexity'' of the output, in analogy to the Kolmogorov complexity), and the expected heat on the input distribution that minimizes entropy production. For universal TMs, we show for both realizations that the thermodynamic complexity of any desired output is bounded by a constant (unlike the conventional Kolmogorov complexity), while the {expected} amount of generated heat is infinite.  
We also show that any computable realization faces a fundamental tradeoff between heat generation, the Kolmogorov complexity of its heat function, and the Kolmogorov complexity of its input-output map. We demonstrate this tradeoff by analyzing the thermodynamics of erasing a long string.
\end{abstract}

\maketitle

\section{Introduction}


%
%
%
%
%
%
%
%
%


The relationship between thermodynamics and information-processing has been an important area of research since at least the 1960s, when Landauer proposed   
%
that any process which erases a bit of information must release at least 
$\ktlntwo$ of heat into its environment~\cite{bril53,bril62,landauer1961irreversibility,szilard1964decrease,zure89a,zure89b,bennett1982thermodynamics,lloyd1989use,dunkel2014thermodynamics,roldan2014universal,lloyd2000ultimate,fredkin1990informational,toffoli1990invertible,leff2014maxwell,maroney2009generalizing,turgut_relations_2009}. %
This research has greatly benefited from the dramatic progress in
nonequilibrium statistical physics in the past few decades, in particular the development of 
trajectory-based and stochastic thermodynamics~\cite{van2013stochastic,van2015ensemble,seifert2012stochastic}.  These developments now permit us to quantify and analyze  heat, work, entropy production and other thermodynamic properties of individual trajectories in far-from-equilibrium systems. They have also    
have led to a much 
deeper understanding of the relationship between thermodynamics and information processing, 
%
%
%
%
 both for information erasure~\cite{berut2012experimental,diana2013finite,zulkowski2014optimal,jun2014high,ciliberto2017experiments,barato2014unifying} and other more elaborate computations~\cite{wiesner2012information,sagawa2012fluctuation,still2012thermodynamics,prokopenko2013thermodynamic,prokopenko2014transfer,roldan2014universal,koski2014experimental,parrondo2015thermodynamics,wolpert_book_2018,wolpert_thermo_comp_review_2019,Boyd2018thesis,strasberg2015thermodynamics,grochow_wolpert_sigact2018,riechers_thermo_comp_book_2018,wolpert_book_2018,ouldridge_thermo_comp_book_2018,strasberg2015thermodynamics}. 
%
%
%

%
%

In this paper we extend this line of research %
by deriving new results on
the thermodynamic costs of performing general computations, as formalized by the notion of \termdef{Turing machines} (TMs). 
A TM is %
an abstraction of a conventional modern computer, which run programs  written in a conventional programming language  
(\textit{C, Python}, etc.)~\cite{sipser2006introduction,hopcroft2000jd,livi08,grunwald2004shannon,arora2009computational,savage1998models}. %
A TM reads  an input string of arbitrary length (a ``program'') and runs until
it produces an output string. %
In the same way that any modern computer can simulate 
other computers (e.g., via an emulator), there exist an important class of TMs called \termdef{universal Turing Machines} (UTMs),
each of which is able to simulate the operation of any other TM. 

TMs are a keystone of the theory of computation~\cite{moore2011nature}, and
touch upon several foundational issues that lie at the intersection of
mathematics and philosophy, such as whether 
$\mathsf{P} = \mathsf{NP}$ and G{\"o}del's incompleteness theorems~\cite{aaronson2013philosophers}. %
Their importance is partly due to the celebrated \emph{Church-Turing thesis}, which postulates that any function that can be computed by a sequence of formal operations can also be computed by some TM~\cite{turing1948intelligent,church1937review,sep-computation-physicalsystems}.  For this reason, in computer science, a function is called \emph{computable} if and only if it can be carried out by a TM~\cite{livi08}.  TMs also play important roles in many facets %
of modern physics. 
For instance, TMs are used to formalize the difference
between easy and hard computational problems in quantum computing~\cite{cubitt2015undecidability,cubitt2012extracting,deutsch1985quantum,benioff1982quantum,nielsen2010quantum}. 
There has also been some speculative, broader-ranging work on whether the foundations
of physics may be restricted by some of the properties of TMs~\cite{barrow2011godel,aaro05}. 
Finally, there has been extensive investigation of the  
\emph{physical Church-Turing thesis}, which states that any function that can be implemented by a physical process 
can also be computed with a TM~\cite{gandy1980church,wolfram1985undecidability,deutsch1985quantum,geroch1986computability,nielsen1997computable,arrighi2012physical,piccinini2011physical,pitowsky1990physical,ziegler2009physically,cubitt2015undecidability,moore1990unpredictability,da1991undecidability,kanter1990undecidability,kieu2003computing,copeland2002hypercomputation}.

%
%
%
%
%
%
%
%
%
%
%
%
%
%
%

%
One of the most important concepts in the theory of TMs
is \emph{Kolmogorov complexity}. 
The Kolmogorov complexity of a string $y$, %
written as $K(y)$, is  the length of the shortest input program 
which causes a UTM to produce $y$ as the output (formal definitions are provided in \cref{sec:AIT}).
%
%
The Kolmogorov complexity of a string $y$ captures the amount of randomness in $y$, because a string with a non-random pattern can be produced
by a short input program. For example, the string containing the first billion digits of $\pi$ can be generated by 
running a very short program, and so has small Kolmogorov complexity. 
In contrast, for a random string $y$ without any patterns, the shortest program that  produces $y$ is a program of the type $\printycmd$, which has about the same length as $y$. 
An important variant of Kolmogorov complexity is the \emph{conditional Kolmogorov complexity} of $y$ given $x$, written $K(y\vert x)$, which is the length of the shortest program which causes a UTM to produce $y$ as output, when the UTM is provided with $x$ as an additional input. 
Kolmogorov and conditional Kolmogorov complexity have many formal connections with entropy and conditional entropy from Shannon's information theory~\cite{grunwald2004shannon}, and are studied in a field called 
\termdef{Algorithmic Information Theory} (AIT)~\cite{livi08,chaitin2004algorithmic}.
%
%
%
%
%
%
%

%
%
%
%
%
%
%
%
%
%
%
%
%
%
%
%
%
%
%
%
%
%
%
%
%
%
%
%
%


In this paper, we combine techniques from AIT and stochastic thermodynamics to analyze the thermodynamics of TMs.  
We imagine a discrete-state physical system %
that is coupled to a heat bath  at temperature $T$, and which evolves under the influence of a driving protocol. 
We identify the initial and final %
 states of the physical system with the logical inputs and outputs of some TM, so that the dynamics over the states of the physical system %
corresponds to a computation performed by the TM. 
We refer to a physical process that is consistent with the laws of thermodynamics and whose dynamics correspond to the input-output map of a TM as a \textit{realization} of that TM. 

%
%
%
%
%
%
%
%
%
%
%
%
%
%
%
%
%
%


We derive numerous results that concern the thermodynamic properties of realizations of TMs.  
The core underlying idea behind these results is that the \emph{logical properties} a given TM (such as the structure of the TM's input-output map, or the Kolmogorov complexity of its inputs and outputs) provide  constraints on the \emph{thermodynamic costs} incurred by realizations of that TM (such as the amount of heat those realizations generate).  Some of our results relate logical properties and thermodynamic costs at the ensemble level (i.e., relative to a probability distribution over computational trajectories of a TM), thereby building on the thermodynamic analysis initiated by Landauer and others.  In addition to these, many of our results also relate logical properties and thermodynamic costs at the level of individual computational trajectories (i.e., individual runs of the TM), which goes beyond most existing research on thermodynamics of computation.  

\subsection{Summary of results}

We investigate three different kinds of thermodynamic costs for a given realization of a TM:

\vspace{5pt}

\noindent (1) The amount of heat that is generated by running the realization of a given (univeral or non-universal) TM on each individual input $x$.
We refer to the map from inputs to their associated heat values %
as the  \textit{heat function} of the TM's realization, and write it as $Q(x)$. %

\vspace{5pt}

\noindent (2) %
The {minimal} amount of heat generated by running the realization of a 
given TM on some individual input that results in a desired output 
$y$. Here we assume that the TM is universal, so that it can in principle produce any output. 
This second cost is a function of the desired {output} $y$, rather than of the input $x$, and can be viewed as a 
thermodynamic analog of conventional Kolmogorov complexity. For this reason, we refer to this cost as  the \textit{thermodynamic complexity} of $y$.
\vspace{5pt}

\noindent (3) The ensemble-level expected heat  $\langle Q \rangle$ generated by the realization of a TM, %
evaluated for the input distribution that minimizes entropy production (EP). For this cost, we again focus on the case of universal TMs.
%
%
\vspace{5pt}

%

%
%
%
%
%

In general, there are many physical processes that are realizations of the same TM, which can have different thermodynamic costs from one another. 
In this paper we consider the above three thermodynamic costs for 
two important types of realizations. 
The first realization we consider, which is called the \textit{coin-flipping} realization, is constructed to be %
thermodynamically
reversible when input programs are sampled from the ``coin-flipping''
distribution $p(x) \propto 2^{-\ell(x)}$, where $\ell(x)$ indicates the length of string $x$. This input distribution arises by feeding random bits into a TM (hence its name) 
and plays a fundamental role in AIT.

We %
show that the heat function of the coin-flipping realization of a given TM is proportional to  $\ell(x)$ minus a ``correction term'' which reflects the logically irreversibility of the input-output map computed by the TM.  Importantly, when the realized TM is a universal TM $\U$, this correction term can be related to the Kolmogorov complexity of the output of $\U$ on input $x$.  In this case, the heat function 
is given by
\begin{align}
\Qcoin(x) =  \bitsToHeatParen{ \ell(\inp) - K(\funcU(x)) } + O(1) ,
\label{eq:3}
\end{align}
where  $\funcU(x)$ indicates the output of $U$ on input $x$, and 
 $O(1)$ indicates equality up to an additive constant independent of $x$ (see \cref{sec:notation} for a formal definition). 
%
Thus, up to an additive constant, the heat 
generated %
by running input $\inp$ on the coin-flipping realization of some UTM $\U$
is proportional to the \emph{excess} length of the input program $\inp$, over and above the length of the shortest program for $\U$ that
produces the same output as $\inp$.

It follows from \cref{eq:3} that if $x$ is the shortest program for $\U$ that produces output $\funcU(x)$, 
then $\Qcoin(x) = O(1)$. %
This means that by running the shortest program $x$ that produces some desired $y$ as output,  
one can produce that $y$ for an amount of heat that is bounded by a constant. 
Thus, the thermodynamic complexity for the coin-flipping realization %
is a {bounded} function, unlike the 
Kolmogorov complexity, which grows arbitrarily large~\cite{livi08}. 
%
%
%
%
%
On the other hand, we also show that when inputs are sampled from the coin-flipping distribution,  
the {expected} heat $\langle Q \rangle$ generated by the coin-flipping realization of a UTM is infinite. This holds even though 
the heat necessary to run the UTM on any given input $x$ is finite. 

The second realization we analyze is inspired by the physical Church-Turing thesis. 
To begin, we refer to a realization of a TM with heat function $\Q$ as a \emph{computable realization} if the function $x \mapsto Q(x)/kT$ is computable 
(i.e., there exist some TM that takes as input any desired $x$ and outputs the corresponding heat value $\Q(x)$ in units of $kT$). 
Under common interpretations of the physical Church-Turing thesis~\cite{deutsch1985quantum,wolfram1985undecidability,geroch1986computability,pitowsky1990physical,nielsen1997computable,sep-computation-physicalsystems}, 
any realization that is \emph{actually} constructable in the real-world must be  computable; in other words, 
a non-computable realization is a hypothetical physical process which does not violate any laws of thermodynamics, but which nonetheless cannot be constructed because of computational constraints.
Motivated by these observations, we define the so-called \emph{dominating realization} of a TM $\T$ to be 
 ``optimal'' in the following sense: the heat it generates on any input $x$ is smaller than the
heat generated by any computable realization of $\T$ on $x$, up to an additive constant which does not depend on $x$.\footnote{
Note that generating minimal heat is different from generating minimal EP. For example, the coin-flipping realization of a TM is thermodynamically
reversible for  the coin-flipping distribution over inputs $x$, and thus generates zero EP when run on inputs sampled from that distribution. However, that does not mean that it generates less heat on any particular input $x$, relative to the heat generated by another realization of the same TM on $x$.} 
The heat function of the dominating realization is proportional to the conditional Kolmogorov complexity of the output given the input,
\begin{align}
\Qopt(x) = \bitsToHeat{ K(x\vert \funcT(x)) } ,
\label{eq:dom1}
\end{align}
where $\funcT(x)$ indicates the output of TM $\T$ on input $x$. 
We show that this heat function is smaller than the heat function $\Q$ of any computable realization of $\T$, 
%
\begin{align}
\Qopt(x) \le \Q(x) + O(1) .
\label{eq:qoptintro}
\end{align}
Note that this result holds whether or not $\T$ is a UTM. 

%
%
%
%
%
%
%
%
%
%
%
%
%
%
%
%

 %

%
%
%
%
%
%
%
%
%
%
%
%
%
%
%
%
%
%
%
%
%
%
%
%
%
%
%
%
%
%
%
%
%
%

%
%
%
%
%
%
%
%

%
For the special case where $\T$ is a UTM, 
we show that for any desired output $y$, the thermodynamic complexity of $y$ under the dominating realization
is bounded by a constant that is independent of $y$, just like for the coin-flipping realization. 
Moreover, %
for the dominating realization 
there is a simple scheme for choosing the 
input $x$ 
that will produce any desired output $\out$ with a bounded amount of heat. 
This differs from the coin-flipping realization, where one must know the shortest program that generates $\out$ in order to produce $\out$ with a bounded amount of heat (in general, finding the shortest program to produce a given output $\out$ is  not computable). 
Finally, we consider the expected heat that is generated by the dominating realization, given some probability distribution over input programs. %
A natural input distribution to consider %
is the 
 one that minimizes the entropy production of the dominating realization. %
As for the coin-flipping realization, we show that the expected heat across inputs sampled from this distribution is infinite. 

%
There are two important caveats concerning the dominating realization.  First, while the dominating realization is better than any computable realization, in the sense of \cref{eq:qoptintro}, it itself is not computable. This is because its heat function is defined in terms of the conditional Kolmogorov complexity, which is not a computable function. Nonetheless, as we discuss below, one can always define a sequence of computable realizations whose heat functions approach $\Qopt$ from above. 
Thus, the dominating realization presents a fundamental bound on the heat generation of computable realizations, and this bound is achievable in the limit.


Second, for a given TM $\T$, \cref{eq:qoptintro} states that the heat generated by  the dominating realization on input $x$, $\Qopt(x)$, 
is 
smaller than the heat generated by any computable realization, $\Q(x)$, up to an additive constant that does not depend on $x$. 
This additive constant, however, can depend on the particular alternative realization of $\T$ that is being compared, i.e., on the choice of comparison heat function $\Q$. 
In fact, depending on the alternative realization, that additive constant can be arbitrarily large and negative. 
This means that for a given TM $\T$ and some particular choice of input program $x$, 
there may exist alternative realizations of $\T$ that generate arbitrarily less heat than the dominating realization.  
It turns out, however, %
that the difference between $\Qopt(x)$ and $\Q(x)$ is upper bounded by 
the sum of the Kolmogorov complexity of the input-output function $\funcT$ and the Kolmogorov complexity of the comparison heat function $\Q$.
Using this result, we show that any computable realization that produces output $y$ from  input $x$ faces a fundamental cost of $K(x\vert y)$, 
which can be paid either by producing a large amount of heat, by computing an input-output map with high complexity, or by having a heat function with high complexity. 


The paper is laid out as follows. In the following subsections, we review relevant prior work and introduce notation. 
In \cref{sec:Background-AIT},   
we define TMs and review some relevant results from AIT. %
In \cref{sec:physback}, we review the basics of statistical physics, and discuss how a TM can be implemented as a physical system.
We present our main results on the coin-flipping and dominating realizations in \cref{sec:EF_coin_flipping} and \cref{sec:EF_optimal}.   
In \cref{sec:tradeoffex}, we demonstrate the tradeoff between heat and complexity by analyzing the thermodynamics of erasing a long string. 
In the last section we discuss potential directions
for future research.

\subsection{Prior work on thermodynamics of TMs}
\label{sec:priorwork}

Some of the earliest work on the thermodynamics of TMs focused on TMs with deterministic and logically reversible
dynamics~\cite{benn73,bennett1989time}. 
Logically reversible TMs can perform computations without generating any heat or entropy production, at the cost of having to store additional information in their output, which logically irreversible TMs do not need to store. 
Due to the thermodynamic costs that would arise in re-initializing that extra stored information, there are some subtleties in calculating the thermodynamic cost of running a ``complete cycle'' of any logically reversible TM~\cite{wolpert_thermo_comp_review_2019}. (See also~\cite{sagawa2014thermodynamic,sagawa2019second} for a discussion of the  relationship between thermodynamic and logical reversibility.) Logically reversible TMs form a special subclass of TMs, and require special definitions of universality~\cite{morita_theory_2017}. In this work, we focus on the thermodynamics of general-purpose TMs, whose computations will generally be logically irreversible. However, we will sometimes also discuss how our results apply in the logically reversible case.

More
recently, \cite{strasberg2015thermodynamics} %
analyzed the thermodynamics of logically reversible TMs %
with stochastic  forward-backward dynamics along a computational trajectory, which 
causes %
the state of the TM %
to become more uncertain with time.\footnote{This kind of ``stochastic TM'' should not be confused with what are called ``nondeterministic
TMs'' or ``probabilistic TMs'' in the computer science literature~\cite{arora2009computational,sipser2006introduction}.} This model incurs %
non-zero entropy production, even though each computational trajectory encodes a logically reversible computation. 
Note that this entropy production could in principle be made arbitrarily small  by driving the TM forward with momentum (e.g., by coupling it to a large flywheel). 
In this work, we %
will ignore possible stochasticity in the progression of a TM along its computational trajectory.

Finally, there has been recent work which interprets the coin-flipping distribution over strings $x$, as defined in \cref{sec:EF_coin_flipping}, as a ``Boltzmann distribution''
induced by the ``energy function'' $\ell(x)$~\cite{baez2012algorithmic}. Doing this allows one 
to formulate a set of equations concerning TMs that
are formal analogs of Maxwell's relations for equilibrium thermodynamic systems.

In our own earlier work, we began to analyze the thermodynamic 
complexity of computing desired outputs, focusing on the coin-flipping realization and a three-tape UTM~\cite{wolpert_arxiv_beyond_bit_erasure_2015}.
We first showed explicitly how to construct a system that is thermodynamically reversible
for the coin-flipping distribution, and then derived the associated heat
function. We showed that 
for this realization,
the minimal amount of heat needed to
compute any given output $\out$ equals the Kolmogorov complexity of $\out$, plus what we characterized as a ``correction term''.
In other, more recent work, we rederived these results
using stochastic thermodynamics and single-tape machines~\cite{wolpert_book_review_chap_2019}.
In this paper, we extend this earlier work
on the coin-flipping realization. For simplicity, we
consider the thermodynamics of systems that implement the entire computation 
of a given UTM in some fixed time interval. (In contrast, our earlier work considered systems 
that implement a given UTM's update function iteratively, taking varying amounts of time
to halt, depending on the input to the UTM.) %
We then go further, and use Levin's Coding theorem to show 
that the thermodynamic complexity of the coin-flipping realization is bounded, 
even though the conventional Kolmogorov complexity function is not.
We also extend this earlier work 
by showing that the coin-flipping realization generates %
infinite {expected} heat when inputs are sampled from the coin-flipping distribution. 

The other main contributions of this paper concern the thermodynamic costs of the dominating realization.
These results are related to a series of ground-breaking  papers begun by
Zurek~\cite{zure89a,zure89b,li1992mathematical,bennett1993thermodynamics,bennett1998information,caves1990entropy,caves1993information,baumeler2019free,zurek1990algorithmic}.  
Those papers 
were generally written before the widespread adoption of trajectory-based analyses of thermodynamics~\cite{van2015ensemble}, and contained a semiformal 
argument
that computing an output string $y$ from an input $x$ has a minimal ``thermodynamic cost'' 
of at least $K(\inp\vert y)$. Even though that semiformal argument is quite different
from our analysis, the same  ``thermodynamic cost'' function also appears in our analysis of the dominating realization. 
We discuss connections between our results and this earlier work in more detail in \cref{sec:tradeoffex}. 

\subsection{Notation}
\label{sec:notation}
\def\ux{\mathbf{u}_X^{[x]}}

We use uppercase letters, such as $X$ and $Y$, to indicate random variables. We use  lowercase letters, like $x$ and $y$, to indicate their outcomes.  We use $p_X$ to indicate a probability distribution over random variable $X$, and $p_{X\vert Y}$ to indicate a conditional probability distribution of random variable $X$ given random variable $Y$. 
 We also use $p_{X|Y=y}$ to indicate the probability distribution of $X$ conditioned on one particular outcome $Y=y$. 
Finally, we use $\supp p_X$ to indicate the support of distribution $p_X$, and notation like $\langle f(X) \rangle_{p_X} = \sum_x p_X(x) f(x)$ to indicate expectations.

A \textit{partial function} $f: A\to B$ is a map from some subset of $A$, which is called the domain of definition of $f$, into $B$. We write $\dom f \subseteq A$ to indicate the domain of definition of $f$, and $\img f := \{ f(a) : a \in \dom f\}$ to indicate the image of $f$. The value of $f(a)$ is undefined for any $a\not\in\dom f$.

For any set $A$, we use $A^*$ to indicate the set of finite strings of elements from $A$. We use $A^\infty$ to indicate the set of infinite strings of elements from $A$. In particular, $\{0,1\}^*$ indicates the set of all finite binary strings. Note that for any finite $A$, $A^*$ is a countably infinite set.

The Kronecker delta is indicated by $\delta(\cdot, \cdot)$.  We sometimes write $\delta_x$ to indicate a delta-function probability distribution over outcome $x$ of random variable $X$, $\delta_x(x') = \delta(x,x')$. 

We use standard asymptotic notation, such as $f(x) = g(x) + O(1)$, which indicates that $|f(x) - g(x)| \le \kappa$ for some $\kappa \in \mathbb{R}$ and all $x$. Similarly, notation like  $f(x) \le g(x) + O(1)$ indicates that $f(x) - g(x) \le \kappa$ for some  $\kappa \in \mathbb{R}$ and all $x$.

\section{Background on Turing Machines and AIT}
\label{sec:Background-AIT}

\subsection{Turing Machines}
\label{sec:TMs}
In its canonical definition, a TM comprises three variables, and a rule for their joint %
dynamics. 
First, there is a \termdef{tape} variable whose state is a semi-infinite string 
$s \in A^\infty$, where $A$ is a finite set of tape symbols which includes a special \emph{blank} symbol. 
Second, there is a \termdef{pointer} variable  $v \in \{1,2,3,\dots\}$, which is interpreted as specifying %
a ``position'' on the tape (i.e.,  an
index into the infinite-dimensional vector $s$). Finally, there is a \termdef{head} variable $h$ whose state belongs to a finite set, 
which includes a specially designated \termdef{start state} and a
specially designated \termdef{halt state}. %

The TM starts with its head in the start state, the pointer set to position 1, and its tape containing some finite string of non-blank symbols, followed by blank symbols. 
The joint state of the tape, pointer, and head evolves over time according to a discrete-time \textit{update function}. 
If during that evolution the head ever enters its halt state, that is interpreted as %
the computation being completed. %
If and when the computation completes, we say that the TM has then \textit{computed} its output,
which is specified by the state of its tape at that time.  Importantly, for some inputs, a TM might never complete its computation, i.e., it may go into an infinite loop and never enter the halt state. 
The operation of a TM is illustrated in a schematic way in \cref{fig:tm}. A more formal definition of a TM and the update function is provided in \cref{app:TMs}.

\begin{figure}
\includegraphics[width=0.556\columnwidth]{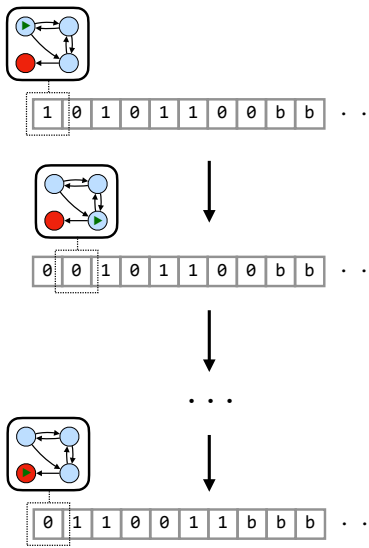}
\caption{A TM performing a computation. The update function is applied over a sequence of steps, causing the  finite-state head  (rounded box, states are colored circles) to move along an infinite tape of symbols ($\texttt{b}$ indicates a special ``blank'' symbol). During each step, the head can read/write the tape symbol in the current position, move left or right along the tape, and change its current state (green triangle). The computation completes if and when the head reaches its halt state (red circle).
\label{fig:tm}
}
\end{figure}

There many other variants of TMs that have been considered in the literature, including ones with multiple tapes and multiple heads. However, all of these variants are computationally equivalent: any computation that can be carried out with a particular TM variant can also be carried out with some TM that possesses a single tape and a single head~\cite{sipser2006introduction,papadimitriou2003computational,wolpert_thermo_comp_review_2019}.

For simplicity of analysis, we make two assumptions about the TMs analyzed in this paper, none of which affect the computational capabilities of the TMs. 
First, we assume that the tape alphabet $A$ contains the binary symbols $0$ and $1$, and that these are the only non-blank symbols present on the tape at the beginning of the computation. 
Second, we assume that any TM 
we consider is designed so that, if and when it reaches a halt state, %
its tape will contain a string from $\BB$ followed by all blank symbols, and the pointer will be set to $1$ (i.e., returned to the start of the tape).  This assumption of a ``standardized'' halt state properly accounts for the thermodynamic costs of running a complete cycle of the TM. For instance, after this standardized halt state is reached, the output of the TM can be moved from the tape onto an off-board storage device and a new input can be moved from another off-board storage device onto the tape, thus preparing the TM to run another program. Importantly, both of these operations can in principle be performed without incurring thermodynamic costs~\cite{wolpert_thermo_comp_review_2019}.

Given the above assumptions, one can represent  
the computation performed by any TM $\T$ as a partial function over the set of finite-length bit strings $\BB$ (see \cref{app:TMs}), which
%
we write as $\funcT : \BB \to \BB$. 
In this notation, $\funcT(x) = y$ indicates that when TM $\T$ is started with input program $x$, it eventually halts and produces the output string $y$.    
Note that  $\funcT$ is a {partial function} because it is undefined for any input $\inp$ for which $\T$ does not eventually halt~\cite{livi08,sipser2006introduction,grunwald2004shannon}. Thus,  $\dom \funcT$ (the domain of definition of $\funcT$) is the set of all input strings on which $\T$ eventually halts, which is sometimes called the ``halting set of $\T$'' in the literature. 

As mentioned in the introduction, a universal TM (UTM) is a TM that can simulate any other TM.
More precisely, given some UTM $\U$ 
and any other TM $\T$, there exists an ``interpreter program'' $\sigma_{\U,\T}$ 
such that for any input $\inp$ of $\T$, $\funcU(\sigma_{\U,\T}, \inp) = \funcT(\inp)$. 
%
%
%
%
%
%
%
%
%
Intuitively, this means that there exists programming languages which are ``universal'',
meaning they can run programs written in any programming language, after
appropriate translation from that other
language.  Note that since $\T$ can itself be a UTM, any UTM can simulate any other UTM. 

Given some partial function $f : \BB \to \BB$ and a TM $\T$, we sometimes say that $\T$ \emph{computes} $f$  if $\funcT=f$ (i.e., $\dom \funcT = \dom f$ and $\funcT(x)=f(x)$ for all $x\in\dom f$). 
We 
say that ``$f$ is computable'' if there exists some TM $\T$ that computes $f$. 
Importantly, there exist functions $\BB \to \BB$ which are \textit{uncomputable}, meaning they cannot be computed by any TM. 
The existence of non-computable functions follows immediately from the fact that there are an uncountable number of functions $\BB \to \BB$, but only a countable number of TMs. 
As an example of an uncomputable function, %
there is no TM which can take any input string $x$ and output a 0 or 1, corresponding to whether or not
$x$ is in the halting set of some given UTM $U$~\cite{livi08,sipser2006introduction,grunwald2004shannon}.

We say that the halting set $\dom \funcT$ is a \emph{prefix-free set} if for any input $x \in \dom \funcT$, there is no other input $x'\in \dom \funcT$ that is a proper prefix of $x$. 
In this paper we %
only consider TMs $\T$ such that  $\dom \funcT$ is prefix-free, which are sometimes called ``prefix TMs'' in the literature. 
Importantly, the set of all prefix TMs
is computationally equivalent to the set of all
TMs: any prefix TM can be simulated by some non-prefix TM and vice-versa. 
However, prefix TMs have many useful mathematical properties, and so
have become conventional in the AIT literature~\cite{livi08}.  See \cref{app:TMs} for a discussion of how prefix TMs can be constructed.
%

%

Above we discussed computable functions from binary strings to binary strings, $\BB \to \BB$.  
It is also possible treat a finite binary string as an encoding
of a pair of binary finite strings. More precisely, assume that along with any TM $\T$, there is a one-to-one \textit{pairing function}
 $\langle a,b \rangle$, which maps pairs of binary strings to single binary strings, and whose %
 image is a prefix-free set. By inverting the pairing function, 
one can uniquely interpret a single binary string %
as a pair of strings. This allows to interpret the domain and/or image of a partial function computed by a TM as a subset of $ \BB\times  \BB$, rather than a subset of $\BB$. 
We will write $\funcT(a,b)$ as shorthand for 
$\funcT(\langle a,b \rangle)$.

It is also possible to interpret a binary string as encoding an integer~\cite{livi08}, or (by inverting the pairing function) as encoding two integers that specify a rational number.  This allows us to formalize the computability of a function from binary strings to integers, $f : \BB \to \mathbb{Z}$, or from binary strings to rationals, $f : \BB \to \mathbb{Q}$.  For a real-valued function $f: \BB \to \mathbb{R}$, we say that $f$ is \termdef{computable} %
if there is a TM that can 
produce an approximation of $f(x)$ accurate to within any desired precision. Formally, $f$ is computable if there exists some TM $\T$ such that $|\funcT(x,n) - f(x)| \le 2^{-n}$ for all $x \in \dom f$ and $n \in \mathbb{N}$. 

\subsection{Algorithmic Information Theory}
\label{sec:AIT}

As mentioned in the introduction, 
the \termdef{Kolmogorov complexity} of any bit string $x \in \BB$ 
is the length of the shortest program %
which leads a given UTM $\U$ to produce $x$ as output. We write this formally as
\begin{align}
K_\U(x) := \min_{z : \funcU(z) = x} \ell(z) \,.
\label{eq:kolmogorov}
\end{align}
The Kolmogorov complexity is unbounded: for any UTM $U$ and 
any finite $\kappa$,  there exists a string $x$ such that $K_U(x) > \kappa$ (this follows from the fact that $\BB$ is an infinite set, while only a finite number of different outputs can be produced by programs of length $\kappa$ or less). 
%
%
%
Moreover, $K_\U$ is an {uncomputable} function. This
implies that if the physical Church-Turing thesis is true, then  no real-world physical 
system can take any desired string $x$ as input and produce the value of $K_\U(x)$ as output.  On the other hand, 
Kolmogorov complexity can be bounded from above,\footnote{For any given $x$, one can compute an improving upper bound on $K_\U(x)$ by running multiple copies of $\U$ in parallel
with different input programs, while keeping track of the length of the shortest program found so far that has halted and produced output $x$~\cite{livi08}. In the limit, this procedure will converge on $K_\U(x)$.} and it is possible to derive many formal results about its properties~\cite{livi08}. 

%

%
%
%
%
%
%
%
%
%
%
%
%
%
%
%
%
%
%
%
%
%
%
%
%
%
%
%
%
%
%
%
%
%
%
%
%
%
%
One can define the Kolmogorov complexity not just for strings but also for computable partial functions. %
Recall from the previous section that given any UTM $\U$ and  TM $\T$, %
there is a corresponding ``interpreter program'' $\sigma_{\U,\T}$, which can be used by $U$ to simulate $\T$
on any input $\inp$.   
The Kolmogorov complexity of a computable function $f$ is defined as the length of the shortest interpreter program for $U$ that simulates a TM that computes $f$:
\begin{align}
K_\U(f) := \min_{\T : \funcT = f} \ell(\sigma_{\U,\T}) .
\end{align}
Similarly, the Kolmogorov complexity of some computable function $f : \BB \to \mathbb{R}$ is given by the length of the shortest interpreter program that approximates $f$ to arbitrary precision. %
 $K_U(f)$ is undefined if $f$ is not computable. 

Above we defined Kolmogorov complexity relative to some particular choice of UTM $\U$. In fact, the choice of $\U$ %
is only relevant up to an additive constant. %
To be precise, for any two UTMs $U$ and $U'$, the ``invariance theorem''~\cite{livi08} states that
\begin{align}
K_{U'}(x) &= K_U(x) + O(1) .\label{eq:invariance}
\end{align}
%
%
Given this result along with the unboundedness of $K_U$,  
for any two UTMs $U$ and $U'$ and any desired $\epsilon > 0$,
\begin{align}
1-\epsilon < {K_U(x)}/{K_{U'}(x)} < 1+\epsilon
\label{eq:2}
\end{align}
for all but a finite number of strings $x$ (out of the infinite set of all possible such strings).  
For many purposes, this allows us to
dispense with specifying the precise UTM $U$ when referring to the Kolmogorov complexity of a string $x$,
and simply write $K(x)$ instead of $K_U(x)$. 

%
%
Finally, the \textit{conditional Kolmogorov complexity} of $x \in \BB$ given $y \in \BB$ is the length of the 
shortest program that, when paired with $y$ and then fed into a UTM $U$, produces 
$x$ as output: 
\begin{align}
K_\U(x|y) := \min_{z : \funcU( z,y ) = x} \ell(z) .
\label{eq:conditional-kolmogorov}
\end{align}
%
%
%
%
%
%
%
%
%
Like regular Kolmogorov complexity, the conditional Kolmogorov complexity is unbounded and uncomputable, though one can derive increasingly tight upper bounds on it.  
In addition, like regular Kolmogorov complexity, the conditional Kolmogorov complexity defined 
relative to two UTMs $U$ and $U'$ differs only up to an additive constant which does not depend on $x$ or $y$~\cite{livi08},
\begin{align}
K_{U'}(x|y) &= K_U(x|y) + O(1) .\label{eq:invariancecond}
\end{align}
Accordingly, for many purposes we can simply write $K(x|y)$, without specifying the precise UTM $\U$.

\section{Background on statistical physics}
\label{sec:physback}

\subsection{Physical setup}

We consider a physical system with a countable state space $\sX$. 
In practice, $\sX$ will often be a ``mesoscopic'' coarse-graining of some underlying phase space, in which case $\sX$ would represent 
the states of the system's ``{information bearing degrees of freedom}''~\cite{bennett2003notes}. 
For simplicity, in this paper we ignore issues raised by coarse-graining,
and treat $\sX$ as the microstates of our system. 

%

%
%
We assume that the system is connected to a work reservoir and a heat bath at temperature $T$. 
The system evolves dynamically under the influence of a driving protocol,
and we are interested in its dynamics over some fixed interval $t\in [0, t_f]$. 

%
%
 %
%
%
%
%
%
%
%
%
%
%
%
%
%

As mentioned in the introduction, 
research in nonequilibrium statistical physics has defined  
%
thermodynamic quantities such as heat, work, and entropy production at the level of individual trajectories of a stochastically-evolving process, so that ensemble averages 
of those measures over
all trajectories obey the usual properties required by conventional statistical physics~\cite{seifert2012stochastic,van2015ensemble}. Adopting this approach, 
we define the \termdef{heat function} $\Q(x)$ as the expected amount of %
heat transferred from our system to the heat bath during the interval $t\in [0, \ft]$, assuming that the system begins in initial state $x$. Following a standard setup in the literature~\cite{Jarzynski2000,esposito2010entropy,sagawa2012thermodynamics,gemmer_quantum_2004}, we assume that the joint Hamiltonian of the system and bath can be written as 
\begin{align}
H^t_X(x) + H_B(b) + H_\mathrm{int}(x,b)
\end{align}
where $H^t_X$ is the time-dependent Hamiltonian of the system, $H_B$ is the bare Hamiltonian of the bath, and $H_\mathrm{int}$ is the interaction Hamiltonian (which is typically very small, reflecting weak-coupling). Regardless of the initial state of the system $x$, the bath is initially taken to be in a Boltzmann distribution $p_B(b)\propto e^{-\heatToEntropy{H_B(b)}}$. 
Let $p_{B\vert x}'$ indicate the final distribution of the bath at $t=\ft$, given that the system began in initial state $x$. The heat function is then given by the  increase of the expected energy of the bath~\cite{Jarzynski2000,esposito2010entropy},
\begin{align}
\label{eq:Qexpr0}
\Q(x) = \langle H_B \rangle_{p_{B\vert x}'} - \langle H_B \rangle_{p_B} \,.
\end{align}
The expectation of $Q(x)$ under any initial distribution $p_X$ then gives the overall expected amount of generated heat averaged across all trajectories, assuming that initial system-bath states are sampled from $p_X(x) p_B(b)$.  This setup can be used to model infinite-sized idealized heat baths (infinite heat capacity, fast equilibration, etc.)  by taking appropriate limits~\cite{Jarzynski2000,esposito2010entropy,sagawa2012thermodynamics,gemmer_quantum_2004}.

%
%

%
%
%

%
%
%

A central quantity of interest in statistical physics is the
(irreversible) \termdef{entropy production} (EP), which reflects the overall increase of entropy in the system and the coupled environment. For a given physical process, let $p_X$ be an initial state distribution at time $t=0$ and let $p_Y$ be the corresponding final state distribution at $t=\ft$.  Then,  
the expected EP is
%
%
%
%
%
%
%
%
%
%
%
%
%
%
\begin{align}
\EP(\pinit) &=  S(\pfin) - S(\pinit) + \heatToEntropy{\langle Q\rangle_{\pinit}} ,
\label{eq:ep1_alternative}
\end{align}
where $S(\cdot)$ indicates the Shannon entropy.\footnote{
\label{foot:ep}
For countably infinite state spaces (e.g., the state spaces of UTMs), the Shannon entropy of both the initial and final distribution can be infinite, making the expression in  \cref{eq:ep1_alternative} ill-defined. In such cases, a finite EP can often be defined by writing \cref{eq:ep1_alternative} as a limit $\EP(\pinit)=\lim_{i\to \infty} \EP(p_i)$, where each $p_i$ has finite support and $\lim_{i\to \infty} p_i= \pinit$.} 
By the second law of thermodynamics, $\EP(\pinit)$ is non-negative for any physically-allowed heat function $Q$ and every initial distribution $\pinit$~\cite{esposito2010entropy}. A physical process is said to be \emph{thermodynamically reversible} if it achieves zero EP.

We say that a physical process is a \emph{realization} of some partial function $f:\sX \to \sX$ if
the conditional probability of the system's ending state given the starting state obeys
\begin{align}
p_{Y|X}(y|x) = \delta(f(x),y) \quad \forall x \in \dom f.
\label{eq:realizes}
\end{align}
The behavior of a realization of $f$ on 
initial states $x \not\in \dom f$ can be arbitrary, as it is not constrained by \cref{eq:realizes}.


%

%
%
%

%
%
%

%
%
%

The following technical result  
 links the logical properties of a partial function $f$ with the heat function of any  realization of that $f$. This result will be central to our analysis, as it will allows us to 
establish thermodynamic constraints on processes that realize TMs.

\begin{restatable}{prop}{propsecond}
\label{prop:second}
Given a countable set $\sX$, let $f : \sX \to \sX$ and $G: \sX \to \mathbb{R}$ be two  partial functions with the same domain of definition. 
The following are equivalent:
\begin{enumerate}
	\item  For all $\pinit$ with $\supp \pinit \subseteq \dom f$,
\begin{align}
\langle \G \rangle_{\pinit} + S(p_{f(X)}) - S(\pinit) \ge 0.
\label{eq:prop1ep}
\end{align}
	\item For all $y \in \img f$,
	\begin{align}
	\sum_{x: f(x)=y} e^{- \G(x)} \le 1.
	\label{eq:prop1ineq}
	\end{align}
	\item There exists a realization of $f$ coupled to a heat bath at temperature $T$, whose heat function $\Q$ obeys
	\begin{align}
	\Q(x)/kT = \G(x) \qquad\forall\;x \in \dom f.
	\label{eq:prop1q}
	\end{align}
\end{enumerate}
\end{restatable}

This proposition is proved in \cref{app:EP-prop}. The proof exploits a useful decomposition of EP into a sum of a conditional Kullback-Leibler divergence term and a non-negative expectation term, which is derived in \cref{app:EP}.


We note two things about \cref{prop:second}. 

First, the remainder of the inequality in \cref{eq:prop1ineq} determines the EP incurred by a realization of $f$. In particular, as we show in \cref{app:EP-prop}, if that inequality is tight for all $y \in \img f$, then the inequality in \cref{eq:prop1ep} is also  tight for some initial  distributions $p_X$. In this case, the realization of $f$, as referenced in \cref{eq:prop1q}, is thermodynamically reversible for those initial $p_X$.   

Second, it is straightforward to generalize the setup described in this section
to consider a 
system connected to multiple thermodynamic reservoirs, instead of a single heat bath~\cite{van2013stochastic}. In the general case, \cref{prop:second} still holds, if the left hand side of \cref{eq:prop1q} is interpreted as 
the amount of entropy increase in all coupled thermodynamic reservoirs, given that the process begins in initial state $x$. 
\cref{eq:prop1q} is a special case of this general formulation, since releasing $Q(x)$ of heat to a bath at temperature $T$ increases the bath's entropy by $Q(x)/kT$. 

\subsection{Realizations of a TM}

\begin{figure}
\includegraphics[width=1\columnwidth]{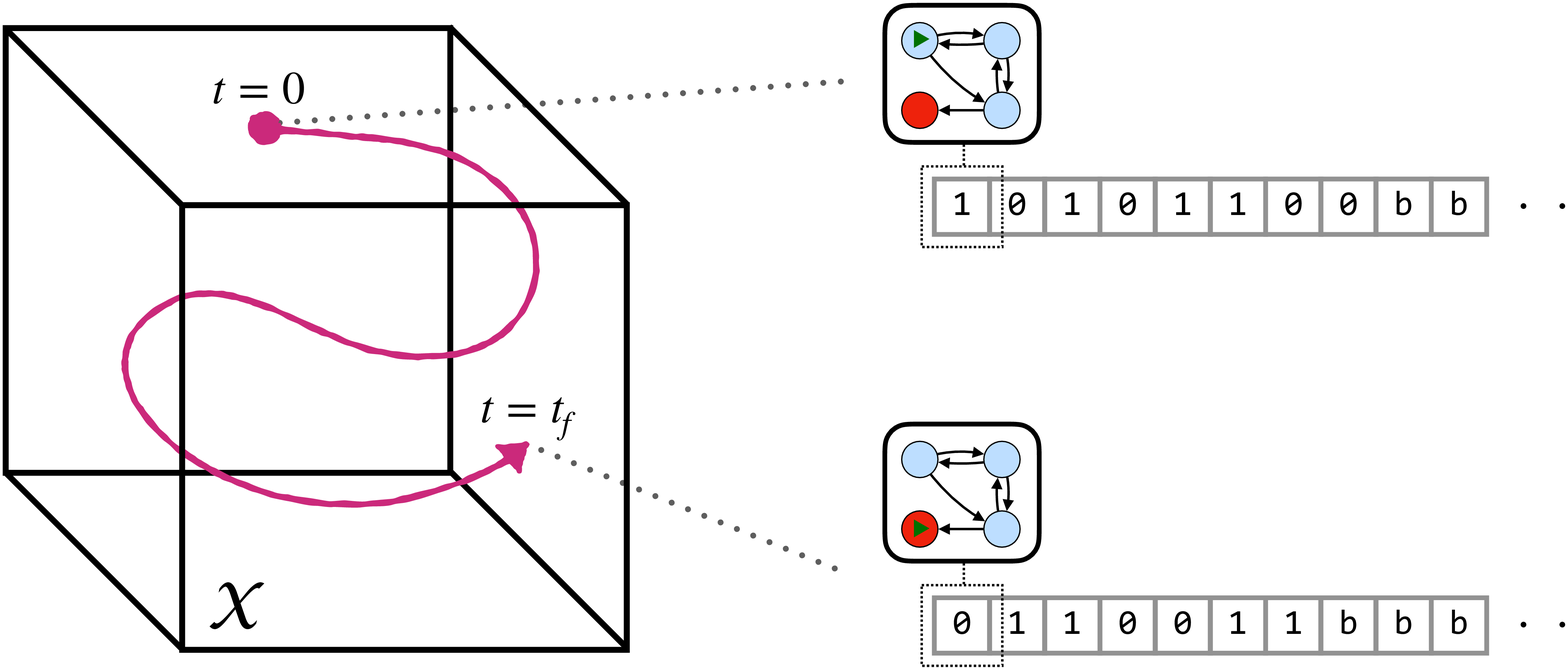}
\caption{A realization of a TM is a physical process over a countable state space $\sX \subseteq \BB$, which maps initial states to final states according to the input-output function of the TM. As a hypothetical example, consider a process that evolves to the final state $\texttt{0110011}$ at $t=\ft$ when started on initial state $\texttt{10101100}$ at $t=0$ (left), as might correspond to a computation performed by the TM (right, see also \cref{fig:tm}).\label{fig:tmphysic}
}
\end{figure}

We briefly describe how a physical process can realize a TM $\T$. 
Without loss of generality, we assume that the countable state space of the physical system $\sX$ can be represented by a set of binary strings, so $\sX \subseteq \BB$.

As described in \cref{sec:TMs} and \cref{app:TMs},  the computation performed by a TM can be formalized as a partial function $\funcT:\BB \to \BB$.  We say that a physical process is a realization of a TM $\T$ if it realizes the partial function $\funcT$, in the sense of \cref{eq:realizes} and \cref{prop:second}. 
Note that this is only possible when $\dom \funcT \cup \img \funcT \subseteq \sX$.   
Note also that there may be physical states $x \in \sX$ that do not belong to $\dom \funcT$.  When the system is initialized with such states at $t=0$, its will undergo some well-defined dynamical evolution.  However, its behavior for such initial states is not constrained by the fact that the system is a realization of the TM, and can be arbitrary (in general, 
the dynamic and thermodynamic properties for such initial $x$ are not our focus). 
The mapping between a TM and a physical system is illustrated in \cref{fig:tmphysic}.




Many TMs, including all UTMs, can have arbitrarily long programs (i.e., unbounded input length), and can take an arbitrary number of steps before halting on any particular input (i.e., unbounded runtime). For such TMs, our formulation appears to assume a physical system that can store a tape of unbounded size, and which can complete an unbounded number of computational steps in a finite time interval $[0,t_f]$, which is not realistic from a physical point of view. In such cases, one can imagine a sequence of realizations, each of which involves manipulating a finite (but growing) tape over a finite (but growing) number of computational steps. Our analysis and results then apply to limit of this sequence, in which the tape size and runtime can be arbitrarily large.


In the following sections, we apply 
\cref{prop:second} with $f = \funcT$ to establish constraints on the heat function $\Q$ of any realization of $\T$. 
We emphasize that in general these constraints do not {fully} determine the heat function of 
any realization of $\T$: there can be many  different realizations of any given TM $\T$, each with different heat
functions and therefore with different thermodynamic properties (see also~\cite{wolpert_thermo_comp_review_2019}).   
In the next sections, 
we analyze the thermodynamics of two particular realizations of a given TM, 
which we call the \emph{coin-flipping realization} and the \emph{dominating realization}.
We work ``backwards'' for each one, first specifying its heat function, then using \cref{prop:nootherconst}
to establish that there is in fact a realization with that heat function, and then analyzing the properties of that heat
function.

Before proceeding, we discuss an important issue concerning the computability properties of realizations of TMs. We say that a realization of a TM $\T$ with heat function $\Q$ is a \emph{computable realization} if the function $\Q(x)/kT$ is computable (i.e., if there exists a TM that can take as input any $x \in \dom \funcT$ and output the value of $\Q(x)/kT$ to arbitrary precision). Some of our results below will rely on particular properties of computable realizations. At the same time, some of the realizations we construct and analyze below will not be computable. Whether such non-computable realizations can \emph{actually} be constructed in the real-world depends on the status of the physical Church-Turing thesis. To see why, imagine that one could construct a non-computable realization of a TM; for example, it might have $\Q(x)/kT = K(x)$, where $K(x)$ is the (non-computable) Kolmogorov complexity function. In that case, one could run the realization on various inputs $x$, use a calorimeter to measure the generated heat in units of $kT$ (i.e., measure $\Q(x)/kT$), and then arrive at the value of $K(x)$. The above procedure would use a physical process to evaluate a non-computable function, thereby violating the physical Church-Turing thesis.

In this paper, we do not take a position on the validity of the physical Church-Turing thesis. Rather, we will explicitly discuss relevant (non-)computability properties of our realizations, as well as how our non-computable realization can be interpreted in light of the physical Church-Turing thesis. It is important to emphasize, however, that even our non-computable realizations are consistent with the laws of thermodynamics, and are well-defined in terms of a sequence of time-varying Hamiltonians and stochastic dynamics (see the construction in the proof of \cref{prop:second}, \cref{app:EP-prop}). Their non-computability arises from the fact that our construction uses various idealizations, such as the ability to apply arbitrary Hamiltonians to the system, which are standard in theoretical statistical physics but which disregard possible \emph{computational constraints} on the set of achievable processes. For example, our construction disregards the fact that, if the physical Church-Turing thesis holds, then it should be impossible to apply non-computable Hamiltonians to the system, such as $H(x) = K(x)$.

\section{coin-flipping realization}
\label{sec:EF_coin_flipping}

We first consider a realization of a  TM $\T$ that achieves zero EP (i.e., is thermodynamically reversible) %
when run on input programs randomly sampled from a particular input distribution. %

To begin, consider the following \textit{coin-flipping} distribution over programs,
which plays an important role in AIT: %
\begin{align}
\label{eq:coin0}
m_X(\inp) := 
\begin{cases}
2^{-\ell(\inp)} & \text{if $x \in \dom \funcT$}\\
0 & \text{otherwise}\,.
\end{cases}
\end{align}
Note that $m_X$ sums to a value less than $1$~\cite{livi08}, therefore $m_X$ is a non-normalized probability distribution. Nonetheless, we refer to it as a ``distribution'', following the convention in the AIT literature.   
%

To understand $m_X$ more concretely, imagine 
that the initial state of the TM's tape is set to a sample of an infinitely long sequence of independent and uniformly distributed bits.  
Then, $m_X(\inp)$ is proportional to the probability that 
$\T$ eventually halts after reading the bit string $\inp$ from the tape.\footnote{For clarity, we omit various technicalities regarding the random process that motivates the coin-flipping distribution. To be precise, this process should be defined in terms of a multi-tape machine, in which one of the tapes is a one-way read-only ``input tape'' (see \cref{app:TMs}). Then, $m_X(\inp)$ is the probability that the multi-tape machine halts after reading the string $x$ from the input tape, assuming the input tape is initialized with an infinitely-long random bit string.} %
Under this hypothetical initialization procedure, the TM will halt on output $y$ with probability 
%
%
%
%
%
%
%
%
%
%
%
%
%
\begin{align}
\uY(y) = \sum_{{x: \funcT(x)=y}} 2^{-\ell(x)}.
\label{eq:u_def}
\end{align}
This output distribution is biased toward strings that can be generated by short  input programs. 
Note that, like $m_X$, this output distribution is not normalized. 
%
%
%
%
%
%
%
%
%
%
%
%
%

We now consider the thermodynamic cost of running a TM on the coin-flipping distribution. We first define a normalized version of the coin-flipping distribution,
\begin{align}
\pcoin(x) := m_X(x) / \Omega_\T,
\label{eq:pcoinX}
\end{align}
where $\Omega_\T := \sum_{x\in\dom \funcT} 2^{-\ell(x)} \le 1$ is a normalization constant (which in AIT is called the ``halting probability''). 
 $\pcoin(x)$ is the probability that a TM halts after running input program $x$, conditioned on the TM halting on \emph{some} input program, given the random initial tape described above.
We also 
define a normalized version of the output distribution,
\begin{align}
\pfincoin(\funcT(x)) := m_Y(\funcT(x)) / \Omega_\T .
\label{eq:pcoinY}
\end{align}

Now consider the associated function
\begin{align}
\G(x) & =- \eplog \pcoin(x) + \eplog \pfincoin(\funcT(x)).
\label{eq:19}
\end{align}
It can be verified that this function satisfies condition 2 of \cref{prop:second}.  Thus, there is at least one realization of $\T$, which we call the 
 \textit{coin-flipping realization}, whose heat function obeys
\begin{align}
\Qcoin(x) = \entropyToHeatParen{- \eplog \pcoin(x) + \eplog \pfincoin(\funcT(x))}.
\label{eq:19q}
\end{align}
By plugging $\Qcoin$ into \cref{eq:ep1_alternative}, we can verify that this realization achieves $\EP(\pcoin) = 0$, meaning that it is thermodynamically reversible  when run on input distribution $\pcoin$.

\cref{eq:19q} can be further simplified by using the definitions of $\pcoin$ and $ \pY^\mathrm{coin}$:
\begin{equation}
\Qcoin(x) 
 = \bitsToHeatParen{ \ell(x) + \log_2 \uY(\funcT(x))} \label{eq:coinres0} \,.
\end{equation}
This establishes the claim in the introduction, that the heat generated under the coin-flipping realization
on input $x$ is proportional to the length of $x$, minus a ``correction term'' $-\log_2 \uY(\funcT(x))$. This correction term is always positive, since $\uY(y)\le 1$ for all $y$. Moreover, it reflects the logical irreversibility of the partial function $\funcT$ on input $x$: it achieves its minimal value of $-\log_2 \Omega$ when $\funcT$ maps all inputs to a single output, and its maximal value of $\ell(x)$ when $\funcT$ is logically reversible on input $x$ (i.e., when $x$ is the only input that produces output $\funcT(x)$).  In the latter logically reversible case, $\Qcoin(x)=0$ for all $x$.

\cref{eq:coinres0} implies that if one wishes to produce some desired output $y \in \img \funcT$ while minimizing heat generation, one should choose the shortest input $x$ such that $\funcT(x)=y$. Loosely speaking, the ``less efficient'' one is in choosing what program to use to compute $y$, the greater the heat
that is expended in that computation.  
Note that this relationship between shorter programs and less heat generation is not a universal feature of all realizations of TMs. It holds for the coin-flipping realization because this realization is explicitly designed to be  thermodynamically-reversible for the coin-flipping input distribution, which has
a ``built-in bias'' for shorter input strings.

An important special case is when the TM of interest is a universal TM.  For any UTM $\U$,  the output distribution in \cref{eq:u_def} is called the \termdef{universal distribution} in AIT. The universal distribution 
 possesses many important mathematical properties and is one of the 
%
cornerstones of AIT~\cite{livi08,chaitin2004algorithmic,chai66,hutter2008algorithmic,hutter2003existence}, and has attracted attention in artificial intelligence~\cite{hutter2004universal,rathmanner2011philosophical,solo64,rissanen1983universal,hutter2003existence,schmidhuber2007new}, foundations of physics~\cite{schmidhuber2000algorithmic,mueller2017law}, 
and statistical physics~\cite{tadaki_generalization_2002,calude_natural_2006,tadaki_statistical_2010,baez2012algorithmic}. 
%
In particular, 
``Levin's Coding Theorem''~\cite{livi08} relates the universal distribution to Kolmogorov complexity,
\begin{align}
- \log_2 \uY(y) = K(y) + O(1) \,.
\label{eq:coding-theorem}
\end{align}
This implies that for a UTM, the ``correction term'' mentioned above is equal to the Kolmogorov complexity of the output, up to an additive constant. 
%

Plugging \cref{eq:coding-theorem} into \cref{eq:coinres0} lets us write the heat function of the coin-flipping realization of a UTM as
%
\begin{equation}
\Qcoin(x) =  \bitsToHeatParen{ \ell(x) - K(\funcU(x)) } + O(1) \,.
\label{eq:coinwork}
\end{equation}
So for a coin-flipping realization of a UTM,
the heat generated on input $x$ reflects how  much the length of $x$ exceeds the shortest program which produces the same output as $x$.





%
%
%
%

These results 
allow us to calculate the thermodynamic complexity of any output string $y$ using the coin-flipping realization of a UTM $\U$, i.e., the minimal heat necessary to generate some desired output $y$:
\begin{align}
\min_{{x : \funcU(x) = y}} \Qcoin(x) & = O(1)\label{eq:coinres2} ,
\end{align} 
where we've used \cref{eq:coinwork} and the fact that $\min_{x:\funcU(x)=y} \ell(x) = K(y)$ by definition. 
Thus, for the coin-flipping realization, the minimal heat 
required by the UTM to compute $\out$ is 
bounded by a constant. %
As emphasized above,
this is a fundamental difference 
between thermodynamic complexity of the coin-flipping realization and 
Kolmogorov complexity, which is unbounded as one varies over $\out$.

%
%
%
However, in order to actually produce a desired output $y$ on a UTM $U$ %
while generating the minimal possible amount of heat, 
one needs to know the shortest program for that $y$. 
Unfortunately, the shortest program for a given output is not computable in general.
In fact, we prove in \cref{app:nootherconst} that there cannot exist a computable function that maps any desired output $y$ to some corresponding input $x$ such that both $\funcU(x)=y$ and the heat is bounded by a constant, $\Qcoin(x)=O(1)$.

We finish by considering the {expected}  
 heat that would be generated by a realization of a UTM $U$ 
if inputs were drawn from the distribution $\pcoin$. 
To begin, rewrite \cref{eq:ep1_alternative}
as
\begin{align}
\langle Q \rangle_{\pcoin} =  \entropyToHeatParen{ S(\pcoin) - S(\pY^\mathrm{coin}) + \EP(\pcoin) } 
\label{eq:new_30}
\end{align}
In \cref{app:infiniteheat}, we show that the difference of entropies on the RHS
of \cref{eq:new_30}  is  infinite.  Since $\EP(\pcoin)$ is always non-negative, any realization of $\U$ %
must, on average, expend an infinite amount of heat to run input programs sampled from $\pcoin$.  This applies to the coin-flipping distribution, for which $\EP(\pcoin)=0$, as well as any other realization.  Note that $\ell(x) \ge \heatToBits{\Qcoin(x)}$ (by \cref{eq:coinres0} and the fact that $m_Y(y) \le 1$ for all $y$), and that 
$\ell(x)$ is a lower bound on the number of steps that a prefix UTM needs to run program $x$ (since it must take at least one step per read-in bit).  Thus, the fact that programs sampled from the coin-flipping distribution have infinite expected heat generation also implies that they have an infinite expected length, and take an infinite expected number of steps before halting.

%
%
%
%
%

%
%
%
%
%
%

We finish by emphasizing that 
EP and expected heat vary in different ways as one changes the initial distribution. 
For example, if we run the coin-flipping realization on input distribution $\pcoin$, then EP is zero while expected heat is infinite. On the other hand, since expected heat is a linear function of the input distribution, minimal expected heat corresponds to a  delta-function input distribution centered on the $x$ that minimizes $\Qcoin(x)$.  However, some simple algebra shows that any such delta-function distribution incurs a strictly positive EP for any UTM.\footnote{Given a UTM and any string $y$, there are many inputs $x$ that result in $\funcU(x)=y$. This means that $\pfincoin(\funcU(x))>\pcoin(x)$ for any $x$, so $\Qcoin(x)>0$ by \cref{eq:19q}.  Thus, for any delta-function distribution $\delta_x$, 
$\EP(\delta_{x})= S(\delta_{\funcU(x)}) - S(\delta_{x}) + Q({x}) = Q({x})>0$, where we've used $S(\delta_{\funcU(x)}) = S(\delta_{x})=0$.}  
Thus, the distribution that minimizes expected heat cannot be the one that minimizes EP.

\section{Dominating realization}
\label{sec:EF_optimal}


%
%
%
%

\subsection{Minimal possible heat function}

We now consider a realization  of a TM 
whose heat function 
is  smaller, up to an additive constant, than the function of any computable realization.

%
%
%
%
%
%
%
%
%
%
%
%
%
%
%
%
%
%
To begin, given any (universal or non-universal) TM $\T$, consider the associated function $G(x)= \bitsToEntropy{ K(x \vert \funcT(x)) }$. 
Note that this conditional Kolmogorov complexity 
can be defined in terms of any desired UTM, with no \textit{a priori} relation to $\T$. %
In \cref{app:dom}, we show that this function $G$ satisfies condition 2 in \cref{prop:second}. Therefore, there must be at least one realization of $\T$,
which we call the \textit{dominating realization}, whose heat function obeys 
\begin{align}
\Qopt(x) =  \bitsToHeat{ K(x \vert \funcT(x)) }.
\label{eq:qoptdef}
\end{align}

%
%
%
%
%
%
%
%
%
%
%

%
Intuitively speaking, the inputs $x$ that generate a large amount of heat 
under the dominating realization of a TM $\T$
are long and incompressible, %
even when given knowledge of their associated outputs $\funcT(x)$.  
An example of such an input 
is a program $x$ that instructs $\T$ to read through a long and 
incompressible bit string and then output nothing, so that $\funcT(x)$ is an empty string (this example is analyzed in more depth below, in \cref{sec:tradeoffex}). 
In contrast, the inputs $x$ that generate little heat under the dominating
realization are those in which  the output provides a large amount of information about the
associated input program.  For instance, if $\T$ is universal, then 
a program $x$ that consists of the instruction $\printycmd$  (represented in some appropriate binary encoding) generates little heat, since $K(\printycmd\vert y) = O(1)$ for any $y$. 
More generally, if $\funcT$ is logically reversible over its domain, then $K(x\vert \funcT(x))=O(1)$ for \textit{all} $x$ in that domain, because one can always reconstruct the input $x$ from the output $\funcT(x)$ by applying $\funcT^{-1}$. 
Thus, in the (logically reversible) case, the heat generated by the dominating realization on any input $x$ is bounded by a constant that doesn't depend on $x$.
%
%

Now consider any alternative computable realization of $\T$ that is coupled to a heat bath at temperature $T$, whose heat function we indicate by $\Q$. 
The assumption of computability means that the 
function $\heatToEntropy{\Q(x)}$ is computable (i.e., there is some TM that, for any desired $x$, can approximate the value of $\Q(x)$ in units of $kT$ to arbitrary precision). 
As we prove in \cref{app:dom}, the heat function of this alternative realization must  obey the following inequality,
\begin{align}
\label{eq:qoptineq}
\Q(x) \ge \Qopt(x) - \bitsToHeat {  K(\heatToEntropy{\Q}) + K(\funcT) } + O(1) ,
\end{align}
where $K(\heatToEntropy{\Q})$ is the Kolmogorov complexity of   
the heat function $\Q$ in units of $kT$, $K(\funcT)$ is the Kolmogorov complexity of the partial function computed by $\T$, and 
 $O(1)$ represents equality up to an additive constant (that does not depend on $x$, $\Q$, or $\T$).   
%
%
%
%
%
%
%
%
%
%
%
%
%
%
%
%
%
%
%

Since neither $K(\heatToEntropy{\Q})$  nor $K(\funcT)$ depends on the input $x$, \cref{eq:qoptineq} implies
$\Q(x) \ge \Qopt(x) + \kconst$ for some constant $\kconst$ that is independent of $x$.  Note though that $\kconst$ 
can depend on $\funcT$ (the partial function being computed) and the alternative realization $Q$, and note also that in principle  this constant may be arbitrarily large and negative. 
%
%
This means that for any {fixed} input $x$, there
may be computable realizations that result in far less heat when run on $x$ than does the dominating realization.  However, this can only occur if $\funcT$ has high complexity (large value of $K(\funcT)$), or if the heat function has high complexity, as reflected by a large value of $K(\heatToEntropy{\Q})$.  This shows that any computable realization  must face a fundamental tradeoff between three different factors: the ``lost'' algorithmic information about the input in the output, the complexity of the input-output map being realized, and the complexity of the heat function.  We explore this tradeoff using an example of erasing a long string in \cref{sec:tradeoffex}.


When the TM under question is universal, then it is guaranteed that there exists some program that can generate any desired output $y$.  This permits us to analyze the thermodynamic complexity of the dominating realization. 
It turns out that, as for the coin-flipping realization, this amount is bounded by a constant:
\begin{align}
\min_{x : \funcU(x) = y} \Qopt(x)  = O(1) \,.
\label{eq:kc}
\end{align}
This minimum is achieved by programs of the form $x = \printycmd$, since these programs achieve $K(x\vert \funcU(x))=O(1)$.  
\cref{eq:kc} also holds if the TM is not a UTM, as long as 
for each each output $y$, there is some $x$  that obeys  $\funcT(x)=y$ and $K(x|\funcT(x)) = O(1)$ (e.g., if $\funcT$ is logically reversible).

Finally, we consider the {expected}
heat that would be generated by running the
dominating realization of a UTM $\U$, assuming that inputs are sampled randomly from some input distribution. 
To parallel the analysis of the coin-flipping realization,
we consider the input distribution which results in minimal EP for 
the dominating realization, which we call $p_X^*$.  In \cref{app:infiniteheat}, we  prove that
the expected heat generated by the dominating realization on the input distribution $p_X^*$ is infinite. 
It is interesting to note that $\ell(x) \ge \Qopt(x)/(kT\ln 2)+ O(1)$ and, as we mentioned above, $\ell(x)$ is a lower bound on the number of steps that a UTM needs to run program $x$.\footnote{We have the inequalities $K(x\vert y)\le K(x) +O(1) \le \ell(x)+O(1)$. The first comes from subadditivity of Kolmogorov complexity~\cite{livi08}, while the second comes from \cref{lem:kbound} in \cref{app:domcoin}.} Thus, the fact that programs sampled from $p_X^*$ have infinite expected heat generation also implies that they have an infinite expected length, and an infinite expected runtime.  
Note that the dominating realization of a UTM will in general incur a strictly positive amount of EP, even when  run on the optimal input distribution $p_X^*$ (see \cref{app:posEP} for details).

%
%
%
%
%
%
%
%
%
%
%
%
%
%
%
%
%
%
%
%
%
%
%
%
%
%
%
%
%
%
%
%

%

%

%

%
%
%
%
%
%
%
%
%
%
%

%
%
%
%

%
%
%
%
%
%
%
%
%
%
%

%
%
%
%


%
%
%

%
%
%
%
%
%
%
%
%
%
%
%

%

%

%
%
%
%
%
%
%
%
%
%

%

%
%
%
%
%
%

%

%
%
%
%
%
%

%
%
%
%
%
%
%
%
%
%
%

\subsection{Practical implications of the dominating realization}
\label{sec:pctt}

Our analysis of the dominating realization uses several abstract computer science concepts, such as the computability of the heat function and its Kolmogorov complexity. 
It is worth making some comments about the real-world
significance of such concepts for the thermodynamics of physical systems.
%
%

%

First, the computability properties of the heat function are entirely separate from the computability properties of the logical map $\funcT$ realized by a physical process. In particular, the heat function can be uncomputable even though $\funcT$ is computable by definition (since $\funcT$ is the partial function implemented by a TM).
On the other hand, 
common interpretations of the physical Church-Turing thesis imply that the heat function of any \textit{actually} constructable real-world physical process
must be computable. This implies that, if the physical Church-Turing holds, the dominating realization generates less heat, up to an additive constant, than any  realization that can actually be constructed in the real-world. 
%
 
%

%
%
%
%
%
%
%
%

%

At the same time, while the dominating realization is better than any computable realization, it is important to note that it itself is not computable. 
This is because the conditional Kolmogorov complexity is not a computable function, i.e., there is no TM that can take as input two strings $x$ and $y$ and output the value of $K(x\vert y)$.  However, this does not necessarily imply that the dominating realization is irrelevant from a practical point of view.  This is because $K(x\vert y)$ is an \emph{upper-semicomputable} function, meaning that 
it is possible to compute an improving sequence of upper bounds that converges on $K(x\vert y)$. Formally, there is a computable function $f$ such that $f(x,y,n)\ge f(x,y,n+1)$ and $\lim_{n\to \infty} f(x,y,n) = K(x\vert y)$.\footnote{This function can be computed by a TM that runs multiple programs in parallel, while keeping track of the shortest program which has halted on input $y$ with output $x$.} 

The upper-semicomputability of $\Qopt$ allows one to approach the performance of $\Qopt$ by constructing a sequence $i=1,2,\dots$ of  realizations of $\funcT$, each with a computable heat function $Q_i$, such that $Q_i$ converge from above on $\Qopt$. 
Each subsequent realization in this sequence is guaranteed to be better (generate less heat) on every input than the previous. 
Moreover, because the heat functions converge on $\Qopt$, by advancing far enough in this sequence one can run any input $x$ with only $\Qopt + \epsilon$ heat for any $\epsilon > 0$. 
An important subtlety, however, is that one cannot compute how far into the sequence to advance so as to  be within $\epsilon$ of $\Qopt$ (if one could compute this, then $\Qopt$ would be computable, and not just upper-semicomputable). 

Finally, while we showed that $\Qopt$ is better than any computable realization in terms of heat generation, we also mentioned that it itself is only upper-semicomputable, not computable.  One might ask if there is some other upper-semicomputable realization (i.e., one whose heat function  can be approached by above) which is even better than $\Qopt$.  It is known that this is not the case: the optimality result of \cref{eq:qoptineq} holds not only for any computable $\Q$, but more generally for any upper-semicomputable $\Q$. 

\subsection{Comparison of coin-flipping and dominating realizations}
\label{sec:comparison}
We finish our discussion of the dominating realization by briefly comparing it to the coin-flipping realization. 

%
%
%
%
%
%
%
%

First, for both dominating and coin-flipping realizations, 
the minimal heat necessary to generate a given output $y$ on a UTM $U$, which we call the {thermodynamic complexity} of the realization, is bounded by a constant that does not depend on $y$. 
There is no \textit{a priori} relationship between those two constants, and in principle it is possible that, for all $y$, the thermodynamic complexity is larger under the dominating realization than the coin-flipping realization, or vice versa. In general, the constants
 will depend on the realized UTM $U$, as well as the 
 UTM used to define the conditional Kolmogorov complexity in \cref{eq:qoptdef} (which does not have to be the same as $U$). 
Second, %
to achieve bounded heat production for output $y$ under the coin-flipping realization, one must know the shortest
program for producing $y$, which is uncomputable. 
In contrast, 
to achieve bounded heat production for output $y$ under the dominating realization, it is sufficient to choose 
an input of the form $\printycmd$. 

%

%

Third, for both realizations, there is an infinite amount of {expected} heat generated, assuming that inputs are sampled from the EP-minimizing distribution.

Fourth, the coin-flipping realization is (by design) thermodynamically-reversible for input distribution $\pcoin$. The dominating realization, on the other hand, is not thermodynamically-reversible for any input distribution (see \cref{app:posEP}). %

Finally, neither the coin-flipping nor the dominating realization of a UTM has a computable heat function.  In fact, the heat function of the coin-flipping realization is not even upper-semicomputable.\footnote{ 
Recall that $\Qcoin(x)=\ell(x) + \log \uY(U(x))$. $\ell(\cdot)$ is
computable while $-\log \uY(\cdot)$ is upper-semicomputable~\cite[Thm.~4.3.3]{livi08}. This implies that $\Qcoin$ is ``lower-semicomputable'', meaning it can be approximated by an improving sequence of computable lower bounds.} 
This means that our results concerning the superiority of the dominating realization do not apply when comparing to the coin-flipping realization, and in particular it is not necessarily the case that $\Qcoin(x) \ge \Qopt(x) + O(1)$.  
Nonetheless, it turns out that for any UTM $U$, the additional heat incurred by the dominating realization on input $x$, beyond that incurred by the coin-flipping realization, is bound by a logarithmic term in the complexity of the output,
\begin{align}
\Qcoin(x) \ge \Qopt(x) -  O(\log K(\funcU(x))).  
\label{eq:logbetter}
\end{align}
(See \cref{app:domcoin} for proof.)  
Such logarithmic correction terms are considered inconsequential in some previous analyses of the thermodynamics of TMs~\cite{zure89a,bennett1993thermodynamics}.

\section{Heat vs. complexity tradeoff}
\label{sec:tradeoffex}

\newcommand{\zerostr}{{\small \text{`\texttt{00...00}'}}}

Our analysis of the dominating realization uncovered a tradeoff between heat and complexity faced by any computable physical process. 
In this section, we illustrate this tradeoff by analyzing the thermodynamics of erasing a long bit string.

As before, consider a physical system with a countable state space, which undergoes driving while coupled to a heat bath at temperature $T$. For notational simplicity, in this section we choose units so that $kT = 1$. Assume that the process realizes some deterministic and computable map from initial to final states, which we indicate generically as $f : \BB \to \BB$.  Now imagine that one observes a single realization of this physical process, in which initial state $x$ is mapped to final state $y = f(x)$.

Since this is a computable realization of $f$, it must obey the dominating realization bound of \cref{eq:qoptineq}.  
Plugging \cref{eq:qoptdef} into that inequality and rearranging gives
\begin{align}
\Q(x)/\ln 2 + K(\Q) + K(f) \ge K(x\vert y) + O(1),
\label{eq:tradeoffbound}
\end{align}
where we've used the assumption that $kT = 1$. 
This shows that there is a fundamental cost of $K(x\vert y)$ that is incurred by any computable realization that maps input $x$ to output $y$.  This fundamental cost can be paid either by generating a lot of heat (large $\Q(x)/\ln 2$), by having a high complexity heat function (large $K(\Q)$), or by realizing a high-complexity input-output function (large $K(f)$).  This tradeoff is illustrated in \cref{fig:tradeoff}.

\begin{figure}
\includegraphics[width=0.8\columnwidth]{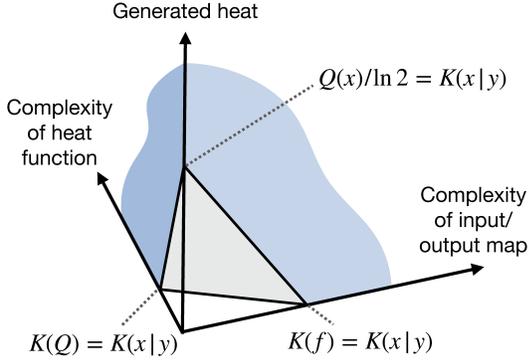}
\caption{Any computable process that realizes a deterministic input-output map $f$ faces a fundamental cost of $K(x\vert y)$ for mapping input $x$ to output $y=f(x)$.  
This cost can be paid through some combination of three different strategies: generating a large amount of heat, having a high complexity heat function, or having a high complexity input-output map $f$.  This tradeoff is illustrated on three axes, with blue indicating the feasible region. 
\label{fig:tradeoff}
}
\end{figure}

We demonstrate this tradeoff using an example of a process that erases a long binary string. 
In this example, $x$ is a long string consisting of $n$ binary digits, while the final state $y$ is a string of $n$ 0s, which we write $\zerostr$.  Assuming $x$ is incompressible (which is true for the vast majority of all strings), the fundamental cost of mapping $x\to y$ is given by $K(x\vert y)=K(x)\approx \ell(x)$ up to logarithmic factors~\cite{livi08}.  Different processes can pay this fundamental cost in different ways, thereby satisfying \cref{eq:tradeoffbound}:

\vspace{5pt}

\noindent (1) A process can generate a lot of heat. For example, in order to erase string $x$, the process can run an erasure map:
\begin{align}
f(x') := \zerostr \qquad\forall x',
\label{eq:erasemap}
\end{align}
while using the dominating implementation. In this case, $\Q(x)/\ln 2 =  K(x\vert y)$ by \cref{eq:qoptdef}.

\vspace{5pt}

\noindent (2) A process can have a high-complexity heat function, so that $K(Q) \ge \ell(x)$.  For example, one can tweak the dominating realization of the erasure map, so that the heat values for input $x$ and the input consisting of all 0s are swapped:
\begin{align*}
 	Q(x') :=\begin{cases}
 	\Qopt(x') & \text{if $x' \not\in\{x,\zerostr\}$} \\
 	\Qopt(\zerostr) & \text{if $x' = x$}\\
\Qopt(x) & \text{if $x' = \zerostr$}
\end{cases}
 	\end{align*}
 One can verify that since $\Qopt$ satisfies condition 2 in \cref{prop:second}, so does this $\Q$. 
 Moreover, this realization generates a small amount of heat when erasing $x$,
 \begin{align*}
 \Q(x) &= \Qopt(\zerostr) \\&= K(\zerostr\vert \zerostr) \approx 0.
 \end{align*}
Note, however, that the long input string $x$ is now ``hard-coded'' into the definition of the heat function $\Q$, leading to a large value of $K(\Q)$.

\vspace{5pt}

\noindent (3) A process can realize a high-complexity input-output map $f$, so that $K(f) \ge K(x\vert y)$.  This strategy could be used, for example, by a process which implements the following logically reversible map:
\begin{align*}
 	f(x') := \begin{cases} x' & \text{if $x'\not \in \{x,\zerostr\}$} \\
 	\zerostr & \text{if $x' = x$}\\
x & \text{if $x' = \zerostr$}
 	\end{cases}
 	\end{align*} 
Since logically reversible function can be carried out without generating heat, it is possible to implement this $f$ while achieving $Q(x')=0$ for all $x'$. In this case, not only does erasing $x$ not generate any heat, $Q(x)=0$, but the heat function has low complexity, $K(Q)\approx 0$.    Now, however, the long input string $x$ is  ``hard-coded'' into the  definition of the input-output map $f$, leading to a large value of $K(f)$.

We finish by noting that in a series of papers by Zurek and others~\cite{zure89a,zure89b,bennett1993thermodynamics,bennett1998information,baumeler2019free,li1992mathematical,caves1990entropy,caves1993information,zurek1990algorithmic},
it was argued that the conditional Kolmogorov complexity $K(x\vert y)$  is ``the minimal thermodynamic cost'' of %
 computing some output $y$ from input $x$.  %
However, most of these early papers
were written before the development of modern nonequilibrium statistical
physics. %
As a result, the arguments in those papers are 
rather informal, which in turn makes it difficult to translate them
in a fully rigorous manner into modern nonequilibrium statistical physics.
(See Sec.\,14.4 in~\cite{wolpert_thermo_comp_review_2019} for one possible translation.)
To give one example of these difficulties, those earlier analyses quantified
the ``thermodynamic cost'' in terms of
the number of physical bits (binary degrees of freedom) that are erased during that computation, independent of
the initial probability distributions over those binary degrees of freedom. However, we now know that
minimal heat generation is given by changes in Shannon entropy, i.e., in terms of statistical bits rather than physical bits.  
Relatedly, these papers led to some proposals that the foundations of statistical physics be changed, so that thermodynamic entropy is identified not only with Shannon entropy, but also a Kolmogorov complexity term \cite{zure89b,livi08}.

%
%
%

In contrast, our analysis is 
grounded in modern nonequilibrium physics, and does not involve any foundational modifications to the definition of thermodynamic entropy. Moreover, it covers some %
issues not considered in earlier analyses.  
In particular, we show that the lower bound of $K(x \vert y)$ is a cost that in general applies %
only to computable realizations (i.e., ones with a computable heat function),  not for all possible realizations, as implied in the earlier papers. The significance of this restriction depends on the legitimacy of the physical Church-Turing thesis. 
Finally, we also demonstrate different ways in which 
one can pay the fundamental cost $K(x \vert y)$: by either generating heat, by having a large Kolmogorov complexity of the heat function $K({\Q})$, or by having a large  Kolmogorov complexity of the input-output map, $K(f)$. 

\section{Discussion}
\label{sec:conclusion}




In this paper we combine Algorithmic Information Theory (AIT) and nonequilibrium statistical physics 
to analyze the thermodynamics of TMs.  We consider a physical process that realizes a deterministic input-output function, representing the computation performed by some TM. 
We derive numerous results concerning two different realizations of TM: a \emph{coin-flipping realization}, which is designed to be thermodynamically reversible when fed with random input bits, and a \emph{dominating realization}, which is designed to generate less heat than any computable realization.

Using our analysis of the dominating realization, we uncover a fundamental tradeoff, faced by any computable realization of a deterministic input-output map, between heat generation, the Kolmogorov complexity of the heat function, and the Kolmogorov complexity of the input-output map. An interesting topic for future research is how the Kolmogorov complexity of the heat function and the input-output map relates to the ``physical complexity'' of the driving process, as commonly understood in physics (e.g., whether the Hamiltonians must have many-body interactions, etc.).

For simplicity, in this paper we represented a TM $\T$ as a physical system  whose dynamics carries out the partial function  $\funcT : \BB \to \BB$ during some finite time interval $[0,\ft]$.  This representation allowed us to abstract away many implementation details of the realization, such as the fact that a TM consists of a separate tape, head, and pointer variables, and that a TM operates in a sequence of discrete steps. 
Essentially, this representation does not distinguish whether the physical process operates via the same sequence of steps as a TM, or simply implements a ``lookup table'' that maps outputs to inputs.

While this representation simplifies our analysis, it provides no guidance on how to actually construct a physical process that realizes a TM in the laboratory, and it leaves implicit some important issues. Alternatively, one could represent a realization of a TM  in a more conventional and ``mechanistic'' way, as a dynamical system over the state of the TM's tape, pointer, and head, which evolves iteratively according to the update function of the TM %
until the head reaches the halt state. 
In contrast to the representation we adopted, this kind of mechanistic representation
could easily be physically constructed, and would correspond more closely to the step-by-step operation of real-world physical computers. 
%
%
%
Moreover, this kind of mechanistic representation could be used to analyze the thermodynamic costs of TMs in a more realistic manner. For example, it could be used to analyze 
how the heat and EP incurred by the TM depends on the number of steps taken.  As another example, it could be used to impose constraints on how the degrees of freedom of the head, tape, and pointer can be coupled together (e.g., via interaction terms of applied Hamiltonians). One might postulate, for instance, that the head of the TM can only interact with  tape locations that are located near the pointer.  These kinds of constraint will generally increase the heat and EP incurred by each step of the TM~\cite{wolpert_thermo_comp_review_2019,circuits2020}. 
These complications concerning the thermodynamics of more mechanistic representations of TMs
are absent from the analysis in this paper, and are topics of future research.

\acknowledgments{\emph{Acknowledgments} --- We would like to thank Josh Grochow, Cris Moore, Daniel Polani, Simon DeDeo, Damian Sowinski, Eric Libby, Sankaran Ramakrishnan, Bernat Corominas-Murtra, and Brendan D. Tracey for many stimulating discussions,  
and the Santa Fe Institute for helping to support this research. This paper was made possible through the support of Grant No. TWCF0079/AB47 from the Templeton World Charity Foundation, Grant No. CHE-1648973 from the U.S. National Science Foundation, Grant No. FQXi-RFP-1622 from the Foundational Questions Institute (FQXi), and Grant No. FQXi-RFP-IPW-1912 from the Foundational Questions Institute (FQXi) and Fetzer Franklin Fund, a donor advised fund of Silicon Valley Community Foundation. 
The opinions expressed in this paper are those of the author and do not necessarily 
reflect the view of Templeton World Charity Foundation.}
%


\bibliographystyle{ieeetr}
\bibliography{thermo_refs_artemy,thermo_refs.main,thermo_refs_2}

\newcommand{\arXiv}[2]{\href{http://arxiv.org/abs/#1}{arXiv:#1 #2}}
\begin{thebibliography}{100}

\bibitem{bril53}
L.~Brillouin, ``Negentropy principle of information,'' {\em Journal of Applied
  Physics}, vol.~24, pp.~1152--1163, 1953.

\bibitem{bril62}
L.~Brillouin, {\em Science and Information Theory}.
\newblock Academic Press, 1962.

\bibitem{landauer1961irreversibility}
R.~Landauer, ``Irreversibility and heat generation in the computing process,''
  {\em IBM journal of research and development}, vol.~5, no.~3, pp.~183--191,
  1961.

\bibitem{szilard1964decrease}
L.~Szilard, ``On the decrease of entropy in a thermodynamic system by the
  intervention of intelligent beings,'' {\em Behavioral Science}, vol.~9,
  no.~4, pp.~301--310, 1964.

\bibitem{zure89a}
W.~H. Zurek, ``Thermodynamic cost of computation, algorithmic complexity and
  the information metric,'' {\em Nature}, vol.~341, pp.~119--124, 1989.

\bibitem{zure89b}
W.~H. Zurek, ``Algorithmic randomness and physical entropy,'' {\em Phys. Rev.
  A}, vol.~40, pp.~4731--4751, Oct 1989.

\bibitem{bennett1982thermodynamics}
C.~H. Bennett, ``The thermodynamics of computation---a review,'' {\em
  International Journal of Theoretical Physics}, vol.~21, no.~12, pp.~905--940,
  1982.

\bibitem{lloyd1989use}
S.~Lloyd, ``Use of mutual information to decrease entropy: Implications for the
  second law of thermodynamics,'' {\em Physical Review A}, vol.~39, no.~10,
  p.~5378, 1989.

\bibitem{dunkel2014thermodynamics}
J.~Dunkel, ``Thermodynamics: Engines and demons,'' {\em Nature Physics},
  vol.~10, no.~6, pp.~409--410, 2014.

\bibitem{roldan2014universal}
{\'E}.~Rold{\'a}n, I.~A. Martinez, J.~M. Parrondo, and D.~Petrov, ``Universal
  features in the energetics of symmetry breaking,'' {\em Nature Physics},
  2014.

\bibitem{lloyd2000ultimate}
S.~Lloyd, ``Ultimate physical limits to computation,'' {\em Nature}, vol.~406,
  no.~6799, pp.~1047--1054, 2000.

\bibitem{fredkin1990informational}
E.~Fredkin, ``An informational process based on reversible universal cellular
  automata,'' {\em Physica D: Nonlinear Phenomena}, vol.~45, no.~1,
  pp.~254--270, 1990.

\bibitem{toffoli1990invertible}
T.~Toffoli and N.~H. Margolus, ``Invertible cellular automata: A review,'' {\em
  Physica D: Nonlinear Phenomena}, vol.~45, no.~1, pp.~229--253, 1990.

\bibitem{leff2014maxwell}
H.~S. Leff and A.~F. Rex, {\em Maxwell's demon: entropy, information,
  computing}.
\newblock Princeton University Press, 2014.

\bibitem{maroney2009generalizing}
O.~Maroney, ``Generalizing landauer's principle,'' {\em Physical Review E},
  vol.~79, no.~3, p.~031105, 2009.

\bibitem{turgut_relations_2009}
S.~Turgut, ``Relations between entropies produced in nondeterministic
  thermodynamic processes,'' {\em Physical Review E}, vol.~79, p.~041102, Apr.
  2009.

\bibitem{van2013stochastic}
C.~Van~den Broeck {\em et~al.}, ``Stochastic thermodynamics: A brief
  introduction,'' {\em Phys. Complex Colloids}, vol.~184, pp.~155--193, 2013.

\bibitem{van2015ensemble}
C.~Van~den Broeck and M.~Esposito, ``Ensemble and trajectory thermodynamics: A
  brief introduction,'' {\em Physica A: Statistical Mechanics and its
  Applications}, vol.~418, pp.~6--16, 2015.

\bibitem{seifert2012stochastic}
U.~Seifert, ``Stochastic thermodynamics, fluctuation theorems and molecular
  machines,'' {\em Reports on Progress in Physics}, vol.~75, no.~12, p.~126001,
  2012.

\bibitem{berut2012experimental}
A.~B{\'e}rut, A.~Arakelyan, A.~Petrosyan, S.~Ciliberto, R.~Dillenschneider, and
  E.~Lutz, ``Experimental verification of landauer's principle linking
  information and thermodynamics,'' {\em Nature}, vol.~483, no.~7388,
  pp.~187--189, 2012.

\bibitem{diana2013finite}
G.~Diana, G.~B. Bagci, and M.~Esposito, ``Finite-time erasing of information
  stored in fermionic bits,'' {\em Physical Review E}, vol.~87, no.~1,
  p.~012111, 2013.

\bibitem{zulkowski2014optimal}
P.~R. Zulkowski and M.~R. DeWeese, ``Optimal finite-time erasure of a classical
  bit,'' {\em Physical Review E}, vol.~89, no.~5, p.~052140, 2014.

\bibitem{jun2014high}
Y.~Jun, M.~Gavrilov, and J.~Bechhoefer, ``High-precision test of landauer's
  principle in a feedback trap,'' {\em Physical review letters}, vol.~113,
  no.~19, p.~190601, 2014.

\bibitem{ciliberto2017experiments}
S.~Ciliberto, ``Experiments in stochastic thermodynamics: Short history and
  perspectives,'' {\em Physical Review X}, vol.~7, no.~2, p.~021051, 2017.

\bibitem{barato2014unifying}
A.~Barato and U.~Seifert, ``Unifying three perspectives on information
  processing in stochastic thermodynamics,'' {\em Physical review letters},
  vol.~112, no.~9, p.~090601, 2014.

\bibitem{wiesner2012information}
K.~Wiesner, M.~Gu, E.~Rieper, and V.~Vedral, ``Information-theoretic lower
  bound on energy cost of stochastic computation,'' {\em Proceedings of the
  Royal Society A: Mathematical, Physical and Engineering Science}, vol.~468,
  no.~2148, pp.~4058--4066, 2012.

\bibitem{sagawa2012fluctuation}
T.~Sagawa and M.~Ueda, ``Fluctuation theorem with information exchange: Role of
  correlations in stochastic thermodynamics,'' {\em Physical review letters},
  vol.~109, no.~18, p.~180602, 2012.

\bibitem{still2012thermodynamics}
S.~Still, D.~A. Sivak, A.~J. Bell, and G.~E. Crooks, ``Thermodynamics of
  prediction,'' {\em Physical review letters}, vol.~109, no.~12, p.~120604,
  2012.

\bibitem{prokopenko2013thermodynamic}
M.~Prokopenko, J.~T. Lizier, and D.~C. Price, ``On thermodynamic interpretation
  of transfer entropy,'' {\em Entropy}, vol.~15, no.~2, pp.~524--543, 2013.

\bibitem{prokopenko2014transfer}
M.~Prokopenko and J.~T. Lizier, ``Transfer entropy and transient limits of
  computation,'' {\em Scientific reports}, vol.~4, p.~5394, 2014.

\bibitem{koski2014experimental}
J.~V. Koski, V.~F. Maisi, T.~Sagawa, and J.~P. Pekola, ``Experimental
  observation of the role of mutual information in the nonequilibrium dynamics
  of a maxwell demon,'' {\em Physical review letters}, vol.~113, no.~3,
  p.~030601, 2014.

\bibitem{parrondo2015thermodynamics}
J.~M. Parrondo, J.~M. Horowitz, and T.~Sagawa, ``Thermodynamics of
  information,'' {\em Nature Physics}, vol.~11, no.~2, pp.~131--139, 2015.

\bibitem{wolpert_book_2018}
D.~H. Wolpert, C.~P. Kempes, P.~Stadler, and J.~Grochow, eds., {\em Energetics
  of computing in life and machines}.
\newblock SFI Press, 2018.

\bibitem{wolpert_thermo_comp_review_2019}
D.~H. Wolpert, ``The stochastic thermodynamics of computation,'' {\em Journal
  of Physics A: Mathematical and Theoretical}, 2019.

\bibitem{Boyd2018thesis}
A.~Boyd, {\em Thermodynamics of Correlations and Structure in Information
  Engines}.
\newblock PhD thesis, Uuniversity of Clifornia Davis, 2018.

\bibitem{strasberg2015thermodynamics}
P.~Strasberg, J.~Cerrillo, G.~Schaller, and T.~Brandes, ``Thermodynamics of
  stochastic turing machines,'' {\em Physical Review E}, vol.~92, no.~4,
  p.~042104, 2015.

\bibitem{grochow_wolpert_sigact2018}
J.~A. Grochow and D.~H. Wolpert, ``Beyond number of bit erasures: New
  complexity questions raisedby recently discovered thermodynamic costs of
  computation,'' {\em ACM SIGACT News}, vol.~49, no.~2, pp.~33--56, 2018.

\bibitem{riechers_thermo_comp_book_2018}
P.~Riechers, ``Transforming metastable memories: The nonequilibrium
  thermodynamics of computation,'' in {\em Energetics of computing in life and
  machines} (D.~H. Wolpert, C.~P. Kempes, P.~Stadler, and J.~Grochow, eds.),
  SFI Press, 2018.

\bibitem{ouldridge_thermo_comp_book_2018}
T.~Ouldridge, R.~Brittain, and P.~Rein Ten~Wolde, ``The power of being
  explicit: demystifying work, heat, and free energy in the physics of
  computation,'' in {\em Energetics of computing in life and machines} (D.~H.
  Wolpert, C.~P. Kempes, P.~Stadler, and J.~Grochow, eds.), SFI Press, 2018.

\bibitem{sipser2006introduction}
M.~Sipser, {\em Introduction to the Theory of Computation}, vol.~2.
\newblock Thomson Course Technology Boston, 2006.

\bibitem{hopcroft2000jd}
J.~E. Hopcroft, R.~Motwani, and J.~Ullman, {\em Introduction to Automata
  Theory, Languages and Computability}.
\newblock Addison-Wesley Longman Publishing Co., Inc., Boston, MA, USA, 2000.

\bibitem{livi08}
M.~Li and P.~Vitanyi, {\em An Introduction to Kolmogorov Complexity and Its
  Applications}.
\newblock Springer, 2008.

\bibitem{grunwald2004shannon}
P.~Grunwald and P.~Vit{\'a}nyi, ``Shannon information and kolmogorov
  complexity,'' {\em arXiv preprint cs/0410002}, 2004.

\bibitem{arora2009computational}
S.~Arora and B.~Barak, {\em Computational complexity: a modern approach}.
\newblock Cambridge University Press, 2009.

\bibitem{savage1998models}
J.~E. Savage, {\em Models of computation}, vol.~136.
\newblock Addison-Wesley Reading, MA, 1998.

\bibitem{moore2011nature}
C.~Moore and S.~Mertens, {\em The nature of computation}.
\newblock Oxford University Press, 2011.

\bibitem{aaronson2013philosophers}
S.~Aaronson, ``Why philosophers should care about computational complexity,''
  in {\em Computability: Turing, G{\"o}del, Church, and Beyond}, pp.~261--327,
  MIT Press, 2013.

\bibitem{turing1948intelligent}
A.~M. Turing, ``Intelligent machinery,'' 1948.

\bibitem{church1937review}
A.~Church, ``Review of turing (1936),'' {\em Journal of Symbolic Logic},
  vol.~2, no.~1, pp.~42--43, 1937.

\bibitem{sep-computation-physicalsystems}
G.~Piccinini, ``Computation in physical systems,'' in {\em The Stanford
  Encyclopedia of Philosophy} (E.~N. Zalta, ed.), Metaphysics Research Lab,
  Stanford University, summer 2017~ed., 2017.

\bibitem{cubitt2015undecidability}
T.~S. Cubitt, D.~Perez-Garcia, and M.~M. Wolf, ``Undecidability of the spectral
  gap,'' {\em Nature}, vol.~528, no.~7581, pp.~207--211, 2015.

\bibitem{cubitt2012extracting}
T.~S. Cubitt, J.~Eisert, and M.~M. Wolf, ``Extracting dynamical equations from
  experimental data is np hard,'' {\em Physical review letters}, vol.~108,
  no.~12, p.~120503, 2012.

\bibitem{deutsch1985quantum}
D.~Deutsch, ``Quantum theory, the church--turing principle and the universal
  quantum computer,'' {\em Proceedings of the Royal Society of London. A.
  Mathematical and Physical Sciences}, vol.~400, no.~1818, pp.~97--117, 1985.

\bibitem{benioff1982quantum}
P.~Benioff, ``Quantum mechanical hamiltonian models of turing machines,'' {\em
  Journal of Statistical Physics}, vol.~29, no.~3, pp.~515--546, 1982.

\bibitem{nielsen2010quantum}
M.~A. Nielsen and I.~L. Chuang, {\em Quantum computation and quantum
  information}.
\newblock Cambridge university press, 2010.

\bibitem{barrow2011godel}
J.~D. Barrow, ``Godel and physics,'' {\em Kurt G{\"o}del and the Foundations of
  Mathematics: Horizons of Truth}, p.~255, 2011.

\bibitem{aaro05}
S.~Aaronson, ``{NP}-complete problems and physical reality.'' quant-ph/0502072,
  2005.

\bibitem{gandy1980church}
R.~Gandy, ``Church's thesis and principles for mechanisms,'' in {\em Studies in
  Logic and the Foundations of Mathematics}, vol.~101, pp.~123--148, Elsevier,
  1980.

\bibitem{wolfram1985undecidability}
S.~Wolfram, ``Undecidability and intractability in theoretical physics,'' {\em
  Physical Review Letters}, vol.~54, no.~8, p.~735, 1985.

\bibitem{geroch1986computability}
R.~Geroch and J.~B. Hartle, ``Computability and physical theories,'' {\em
  Foundations of Physics}, vol.~16, no.~6, pp.~533--550, 1986.

\bibitem{nielsen1997computable}
M.~A. Nielsen, ``Computable functions, quantum measurements, and quantum
  dynamics,'' {\em Physical Review Letters}, vol.~79, no.~15, p.~2915, 1997.

\bibitem{arrighi2012physical}
P.~Arrighi and G.~Dowek, ``The physical church-turing thesis and the principles
  of quantum theory,'' {\em International Journal of Foundations of Computer
  Science}, vol.~23, no.~05, pp.~1131--1145, 2012.

\bibitem{piccinini2011physical}
G.~Piccinini, ``The physical church--turing thesis: Modest or bold?,'' {\em The
  British Journal for the Philosophy of Science}, vol.~62, no.~4, pp.~733--769,
  2011.

\bibitem{pitowsky1990physical}
I.~Pitowsky, ``The physical church thesis and physical computational
  complexity,'' {\em Iyyun: The Jerusalem Philosophical Quarterly}, pp.~81--99,
  1990.

\bibitem{ziegler2009physically}
M.~Ziegler, ``Physically-relativized church--turing hypotheses: Physical
  foundations of computing and complexity theory of computational physics,''
  {\em Applied Mathematics and Computation}, vol.~215, no.~4, pp.~1431--1447,
  2009.

\bibitem{moore1990unpredictability}
C.~Moore, ``Unpredictability and undecidability in dynamical systems,'' {\em
  Physical Review Letters}, vol.~64, no.~20, p.~2354, 1990.

\bibitem{da1991undecidability}
N.~C. da~Costa and F.~A. Doria, ``Undecidability and incompleteness in
  classical mechanics,'' {\em International Journal of Theoretical Physics},
  vol.~30, no.~8, pp.~1041--1073, 1991.

\bibitem{kanter1990undecidability}
I.~Kanter, ``Undecidability principle and the uncertainty principle even for
  classical systems,'' {\em Physical Review Letters}, vol.~64, no.~4, p.~332,
  1990.

\bibitem{kieu2003computing}
T.~D. Kieu, ``Computing the non-computable,'' {\em Contemporary Physics},
  vol.~44, no.~1, pp.~51--71, 2003.

\bibitem{copeland2002hypercomputation}
B.~J. Copeland, ``Hypercomputation,'' {\em Minds and machines}, vol.~12, no.~4,
  pp.~461--502, 2002.

\bibitem{chaitin2004algorithmic}
G.~J. Chaitin, {\em Algorithmic information theory}, vol.~1.
\newblock Cambridge University Press, 2004.

\bibitem{benn73}
C.~Bennett {\em IBM Journal of Research and Development}, vol.~17,
  pp.~525--532, 1973.

\bibitem{bennett1989time}
C.~H. Bennett, ``Time/space trade-offs for reversible computation,'' {\em SIAM
  Journal on Computing}, vol.~18, no.~4, pp.~766--776, 1989.

\bibitem{sagawa2014thermodynamic}
T.~Sagawa, ``Thermodynamic and logical reversibilities revisited,'' {\em
  Journal of Statistical Mechanics: Theory and Experiment}, vol.~2014, no.~3,
  p.~P03025, 2014.

\bibitem{sagawa2019second}
T.~Sagawa, ``Second law, entropy production, and reversibility in
  thermodynamics of information,'' in {\em Energy Limits in Computation},
  pp.~101--139, Springer, 2019.

\bibitem{morita_theory_2017}
K.~Morita, {\em Theory of {Reversible} {Computing}}.
\newblock Monographs in {Theoretical} {Computer} {Science}. {An} {EATCS}
  {Series}, Tokyo: Springer Japan, 2017.

\bibitem{baez2012algorithmic}
J.~Baez and M.~Stay, ``Algorithmic thermodynamics,'' {\em Mathematical
  Structures in Computer Science}, vol.~22, no.~05, pp.~771--787, 2012.

\bibitem{wolpert_arxiv_beyond_bit_erasure_2015}
D.~H. Wolpert, ``Extending {L}andauer's bound from bit erasure to arbitrary
  computation.'' arXiv:1508.05319 [cond-mat.stat-mech], 2015.

\bibitem{wolpert_book_review_chap_2019}
D.~H. Wolpert, ``Overview of information theory, computer science theory, and
  stochastic thermodynamics for thermodynamics of computation,'' in {\em
  Energetics of computing in life and machines} (D.~H. Wolpert, C.~P. Kempes,
  P.~Stadler, and J.~Grochow, eds.), SFI Press, 2019.

\bibitem{zurek1990algorithmic}
W.~H. Zurek, ``Algorithmic information content, church-turing thesis, physical
  entropy, and maxwell's demon,'' tech. rep., Los Alamos National Lab., NM
  (USA), 1990.

\bibitem{li1992mathematical}
M.~Li, {\em Mathematical theory of thermodynamics of computation}.
\newblock Citeseer.

\bibitem{bennett1993thermodynamics}
C.~H. Bennett, P.~G{\'a}cs, M.~Li, P.~Vit{\'a}nyi, and W.~H. Zurek,
  ``Thermodynamics of computation and information distance,'' in {\em
  Proceedings of the twenty-fifth annual ACM symposium on Theory of computing},
  pp.~21--30, ACM, 1993.

\bibitem{bennett1998information}
C.~H. Bennett, P.~G{\'a}cs, M.~Li, P.~M. Vit{\'a}nyi, and W.~H. Zurek,
  ``Information distance,'' {\em Information Theory, IEEE Transactions on},
  vol.~44, no.~4, pp.~1407--1423, 1998.

\bibitem{caves1990entropy}
C.~M. Caves, ``Entropy and information: How much information is needed to
  assign a probability,'' {\em Complexity, Entropy and the Physics of
  Information}, pp.~91--115, 1990.

\bibitem{caves1993information}
C.~M. Caves, ``Information and entropy,'' {\em Physical Review E}, vol.~47,
  no.~6, p.~4010, 1993.

\bibitem{baumeler2019free}
{\"A}.~Baumeler and S.~Wolf, ``Free energy of a general computation,'' {\em
  Physical Review E}, vol.~100, no.~5, p.~052115, 2019.

\bibitem{papadimitriou2003computational}
C.~H. Papadimitriou, {\em Computational complexity}.
\newblock John Wiley and Sons Ltd., 2003.

\bibitem{bennett2003notes}
C.~H. Bennett, ``Notes on landauer's principle, reversible computation, and
  maxwell's demon,'' {\em Studies In History and Philosophy of Science Part B:
  Studies In History and Philosophy of Modern Physics}, vol.~34, no.~3,
  pp.~501--510, 2003.

\bibitem{Jarzynski2000}
C.~Jarzynski, ``Hamiltonian derivation of a detailed fluctuation theorem,''
  {\em Journal of Statistical Physics}, vol.~98, pp.~77--102, Jan 2000.

\bibitem{esposito2010entropy}
M.~Esposito, K.~Lindenberg, and C.~Van~den Broeck, ``Entropy production as
  correlation between system and reservoir,'' {\em New Journal of Physics},
  vol.~12, no.~1, p.~013013, 2010.

\bibitem{sagawa2012thermodynamics}
T.~Sagawa, ``Thermodynamics of information processing in small systems,'' {\em
  Progress of theoretical physics}, vol.~127, no.~1, pp.~1--56, 2012.

\bibitem{gemmer_quantum_2004}
J.~Gemmer, M.~Michel, and G.~Mahler, {\em Quantum {Thermodynamics}: {Emergence}
  of {Thermodynamic} {Behavior} {Within} {Composite} {Quantum} {Systems}},
  vol.~657 of {\em Lecture {Notes} in {Physics}}.
\newblock Berlin, Heidelberg: Springer Berlin Heidelberg, 2004.

\bibitem{hutter2004universal}
M.~Hutter, {\em Universal artificial intelligence: Sequential decisions based
  on algorithmic probability}.
\newblock Springer Science \& Business Media, 2004.

\bibitem{rathmanner2011philosophical}
S.~Rathmanner and M.~Hutter, ``A philosophical treatise of universal
  induction,'' {\em Entropy}, vol.~13, no.~6, pp.~1076--1136, 2011.

\bibitem{solo64}
R.~Solomonoff {\em Information and Control}, vol.~7, 1964.

\bibitem{rissanen1983universal}
J.~Rissanen, ``A universal prior for integers and estimation by minimum
  description length,'' {\em The Annals of statistics}, pp.~416--431, 1983.

\bibitem{hutter2003existence}
M.~Hutter, ``On the existence and convergence of computable universal priors,''
  in {\em Algorithmic Learning Theory: 14th International Conference, ALT 2003,
  Sapporo, Japan, October 17-19, 2003, Proceedings}, vol.~2842, p.~298,
  Springer, 2003.

\bibitem{schmidhuber2007new}
J.~Schmidhuber, ``The new ai: General \& sound \& relevant for physics,'' in
  {\em Artificial General Intelligence}, pp.~175--198, Springer, 2007.

\bibitem{schmidhuber2000algorithmic}
J.~Schmidhuber, ``Algorithmic theories of everything,'' {\em arXiv preprint
  quant-ph/0011122}, 2000.

\bibitem{mueller2017law}
M.~P. Mueller, ``Law without law: from observer states to physics via
  algorithmic information theory,'' {\em arXiv preprint arXiv:1712.01826},
  2017.

\bibitem{tadaki_generalization_2002}
K.~Tadaki, ``A generalization of {Chaitin}'s halting probability $\omega$ and
  halting self-similar sets,'' {\em Hokkaido Mathematical Journal}, vol.~31,
  pp.~219--253, Feb. 2002.

\bibitem{calude_natural_2006}
C.~S. Calude and M.~A. Stay, ``Natural halting probabilities, partial
  randomness, and zeta functions,'' {\em Information and Computation},
  vol.~204, pp.~1718--1739, Nov. 2006.

\bibitem{tadaki_statistical_2010}
K.~Tadaki, ``A statistical mechanical interpretation of algorithmic information
  theory: {Total} statistical mechanical interpretation based on physical
  argument,'' {\em Journal of Physics: Conference Series}, vol.~201, p.~012006,
  Dec. 2010.

\bibitem{chai66}
G.~Chaitin {\em Journal of the Association of Computational Machinery},
  vol.~13, p.~547, 1966.

\bibitem{hutter2008algorithmic}
M.~Hutter, ``Algorithmic complexity,'' {\em Scholarpedia}, vol.~3, no.~1,
  p.~2573, 2008.

\bibitem{zurek1990complexity}
W.~H. Zurek, ed., {\em Complexity, entropy and the physics of information}.
\newblock Addison-Wesley, 1990.

\bibitem{circuits2020}
D.~Wolpert and A.~Kolchinsky, ``The thermodynamics of computing with
  circuits,'' {\em New Journal of Physics}, 2020.

\bibitem{kolchinsky2016dependence}
A.~Kolchinsky and D.~H. Wolpert, ``Dependence of dissipation on the initial
  distribution over states,'' {\em Journal of Statistical Mechanics: Theory and
  Experiment}, p.~083202, 2017.

\bibitem{cover_elements_2012}
T.~M. Cover and J.~A. Thomas, {\em Elements of information theory}.
\newblock John Wiley \& Sons, 2012.

\bibitem{wolpert_spacetime_2019}
D.~H. Wolpert, A.~Kolchinsky, and J.~A. Owen, ``A space/time tradeoff for
  implementing a function with master equation dynamics,'' {\em Nature
  Communications}, 2019.

\bibitem{deffner2013information}
S.~Deffner and C.~Jarzynski, ``Information processing and the second law of
  thermodynamics: An inclusive, hamiltonian approach,'' {\em Physical Review
  X}, vol.~3, no.~4, p.~041003, 2013.

\bibitem{vitanyi2013conditional}
P.~M. Vit{\'a}nyi, ``Conditional kolmogorov complexity and universal
  probability,'' {\em Theoretical Computer Science}, vol.~501, pp.~93--100,
  2013.

\end{thebibliography}

\appendix
\clearpage

\section{Models of single-tape TMs}
\label{app:TMs}

In this appendix we present a formal definition of a single-tape TM. 


%
%
%

%
%
%

%
%
%
%
%
%
%
%
%

\newcommand{\updatefunc}{f}
In \cref{sec:TMs}, we define the state of a TM as being composed of a tape state $s \in A^\infty$, a pointer state $v \in \mathbb{N}$, and head state $h \in H$. Here, $A$ is a finite alphabet of tape symbols which includes a special ``blank'' symbol, while $H$ is a finite set of head states which includes a special ``start'' head state and a special ``halt'' head state.  
Any particular value of the triple  $(s, v, h)$ is called 
an \termdef{instantaneous description} (ID) of the TM. The
dynamics of a particular
TM is given by iteratively applying an \termdef{update function} $\updatefunc$ to the ID,%
\begin{align}
\updatefunc : (s,v, h) \mapsto (s', v', h') \,. 
\end{align}
Following standard definitions, we assume that $\updatefunc(s,v, h)$ only depends on $(s(v), h)$, i.e., the next ID of the TM
can only depend on the current state of the head and the current 
contents of the tape $s$ at position $v$. We also assume that the new value of the pointer $v'$ does not differ by more than $1$
from $v$, and that the tape state $s'$ be identical to the tape state $s$ at all positions, except possibly position $v$. By iteratively applying $\updatefunc$, the head moves back and forth along 
the tape, while both changing its state 
as well as reading and writing symbols onto the tape at its current position.

At the beginning of a computation, the state of the TM must be a \emph{valid initial ID}, 
meaning that the head $h$ is in 
the start state, the pointer is set to $v=1$, and the tape $s$ consists 
of finite string of non-blank symbols, followed by an infinite sequence of blank symbols. 
The TM then visits a sequence of IDs by iteratively applying the update function $\updatefunc$. The TM stops if the head ever reaches the halt state (i.e., any ID where the head in the halt state is a
fixed point of $\updatefunc$). In general, there can be valid initial IDs for which the TM never halts.

For simplicity, we assume that $0$ and $1$ are elements of the alphabet $A$, and that the non-blank finite string at the beginning of the initial tape state is some $x\in \BB$.  
In addition, we assume that if the head of the TM reaches a halt state after starting from some valid initial ID, %
then at that time the pointer is set  to $1$ and the final tape state begins with some $y\in\BB$, followed by  blank symbols. 
In that case, we refer to the string $x \in \BB$ as the \termdef{input} or \termdef{program} for the TM, and the corresponding string $y\in\BB$ as the 
\termdef{output} of the TM for program $x$.

Given these assumption, we can represent the overall computation performed by a TM $\T$ as a partial function $\funcT : \BB \to \BB$. 
Here, $\funcT(x)=y$ indicates that when the TM is initialized with its tape containing $x$ followed by an infinite sequence of blank symbols, then it will halt with its tape containing $y$ followed by an infinite sequence of blank symbols.  If the TM does not halt for some particular initial tape state $x$, then the value of $\funcT(x)$ is undefined (for this reason, in general $\funcT$ is a partial function). 
When
we talk about a realization of a TM $\T$ in the main text, we refer to a physical
process over a countable state space, whose dynamics from initial states to final states can be mapped onto the partial function $\funcT$ implemented by some TM $\T$.

As we mention in  the main text, we assume that any TM under consideration is a prefix TM, meaning that it has a prefix-free halting set. 
Prefix TMs are typically TMs with multiple tapes, where one of the tapes is a read-only input tape that is read left-to-right~\cite{livi08}.  If this kind of multi-tape machine halts after reading some string $x$ from the input tape, it means that the machine did not halt after reading some string $x'$ on the input tape which is a strict prefix of $x$ (otherwise, it would never get to read-in all of $x$), thereby guaranteeing the prefix property. For simplicity, however, in this paper we assume that the prefix TM is single-tape.  This can be done without loss of generality, as it is always possible to transform 
a prefix TM with multiple tapes into an equivalent single-tape prefix
TM, using any of the conventional techniques for transforming between
multi-tape and single-tape TMs (see \cite[Thm.~2.1]{papadimitriou2003computational} and \cite{sipser2006introduction} for details). Note that these techniques may involve adding additional symbols to the tape alphabet $A$, which may be used at intermediate steps of the computation.




%
%
%
%
%
%
%
%
%
%
%

%
%
%
%
%
%
%
%

%
%
%
%
%
%

\section{Decomposition of entropy production}
\label{app:EP}

In this appendix, we derive a useful decomposition of the EP incurred by a realization of a deterministic input-output function.  We also relate this decomposition to our previous work, which analyzed the dependence of EP on the initial distribution of a process~\cite{kolchinsky2016dependence,wolpert_thermo_comp_review_2019,circuits2020}. 

%

Consider some physical process that realizes the function $f : \sX \to \sX$, in the sense of \cref{eq:realizes}. 
Then, the conditional distribution of an initial state 
$x \in \dom f$ given final state $f(x)$ can be written as 
\begin{align}
\pXgY(x \vert f(x)) := \dfrac{\px}{\sum_{x':f(x')=f(x)}p_X(x')}  .
\label{eq:appdefpxy}
\end{align}
We use this expression to rewrite the EP from \cref{eq:ep1_alternative} as%
\begin{align}
& \EP(\pinit) = %
\label{eq:appep0}
 \\
& \quad \sum_{x} \px \bigg[\eplog \frac{ \pXgY(x \vert f(x))} { e^{-\heatToEntropy{\Q(x)} - \logZ(f(X))}} 
            - \logZ(f(x))\bigg],  \nonumber
\end{align}
where  %
we have defined 
\begin{align}
\label{eq:zdef}
Z(y) :=\sum_{x: f(x)=y} \boltz{{\Q(x)}}.
\end{align}
Now, define the following conditional distribution,
\begin{align}
\priorXgY(x\vert f(x)) :=  \boltz{\Q(x)} - \logZ(f(x)) .
\label{eq:mudef}
\end{align}
Using this definition, we can 
further rewrite \cref{eq:appep0} as
\begin{align}
\EP(\pinit) & = D(\pXgY \Vert \priorXgY) - \langle \logZ(f(x)) \rangle_{\pinit},
\label{eq:decomp2} 
\end{align}
where $D(\pXgY \Vert \priorXgY)$ indicates the conditional KL divergence between the conditional distribution $\pXgY$ and 
$\priorXgY$~\cite{cover_elements_2012}.

As we show below in \cref{eq:appres}, $-\logZ(f(x)) \ge 0$ for all $x$. Thus, \cref{eq:decomp2}  implies  $\EP(\pinit) \ge D(\pXgY \Vert \priorXgY)$.  Note that this lower bound is non-negative, and vanishes whenever $\pXgY = \priorXgY$. This means that $\priorXgY$, as defined in \cref{eq:mudef}, encodes that  conditional probability of  inputs $x$ given outputs $f(x)$ that achieves minimal EP for a realization of $f$ with heat function $Q$. 

In our previous work, we have sometimes referred to the conditional KL divergence in \cref{eq:decomp2}  as \emph{mismatch cost}. Using the chain rule for KL divergence, we write mismatch cost as
\begin{align*}
D(\pXgY \Vert \priorXgY) = D(\pinit \Vert \prior_X) - D(p_{f(X)} \Vert w_{f(X)}) \,,
\end{align*}
where $\prior_{f(X)}(y) = \sum_{x : f(x) = y} \prior_X(x)$, while $\prior_X(x)$ is any distribution that obeys
\[
\prior_X(x) / \prior_X(x') = e^{\heatToEntropy{[\Q(x') - \Q(x)]}} \quad \forall x,x':f(x)=f(x').
\] 
In our previous, we referred to the distribution $\prior_X(x)$ as a \textit{prior}. (This term was originally motivated
by a Bayesian interpretation of EP~\cite{kolchinsky2016dependence}.) 
As long as $|\img f| > 1$, there are an infinite number of
priors for any given $\priorXgY$, since the relative probabilities of any pair $x, x'$ with $f(x)\ne f(x')$ are unconstrained.
%

In our previous work~\cite{wolpert_thermo_comp_review_2019,circuits2020}, we referred to the term
$-\langle \logZ (f(X)) \rangle_{\pinit}$ in \cref{eq:decomp2} as the \emph{residual EP}. 
Observe that for any $y \in \img f$,
\begin{align}
\EP(w_{X|f(X)=y})& = D(w_{X|f(X)=y} \Vert w_{X|f(X)})  -\logZ(y) \nonumber \\
& = -\logZ(y) .\label{eq:appres}
\end{align}
Since $\EP(w_{X|f(X)=y}) \ge 0$ by the second law, 
$-\logZ(y)$ is non-negative for all $y \in \img f$ and therefore residual EP is always non-negative.  Note also that the residual EP is an expectation under $\pinit$, thus it is linear in $\pinit$.  
In fact, it only depends on the probabilities assigned to each output $p_{f(X)}(y)$, not the conditional distribution of inputs corresponding to each output. 
In our other work~\cite{circuits2020}, we've  sometimes called the indexed set $\{-\logZ(y)\}_{y}$
the \textit{residual EP parameter}.  

Finally, define an \textit{island} of $f$
as a pre-image $f^{-1}(y)$ for some $y$, with $L(f)$ the set of all
islands of $U$.  We can rewrite \cref{eq:decomp2} as
\begin{align*}
\EP(\pinit) = \sum_{\mathclap{c \in L(f)}} p(c)\left[ D\left(p_{X|X\in c} \Vert w_{X|X\in c} \right) -  \logZ(f(c)) \right],
\end{align*}
where $p(c)=\sum_{x\in c} p_X(x)$. 
Intuitively, this expression shows that any realization of the function $f$ can be thought of a
set of (island-indexed) ``parallel'' processes,
operating independently of one another on non-overlapping subsets of $\sX$,
each generating EP given by the associated mismatch cost and residual EP.

This form of mismatch cost, residual EP, and island decomposition was introduced  in~\cite{kolchinsky2016dependence,wolpert_thermo_comp_review_2019,circuits2020}. It holds 
even in the general case of non-deterministic dynamics, with an appropriate (more general)
definition of the prior $\prior_X$ and the island decomposition.  
However, that previous work on  mismatch cost and residual EP 
assumed finite state spaces. The derivation presented above does not have that restriction.

\def\SetN{\mathcal{A}_n}

\section{Proof of \cref{prop:second}}
\label{app:EP-prop}

The following proof will make use of the decomposition of EP derived in \cref{app:EP}.

\propsecond*
\begin{proof}
Note that condition 1 follows from condition 3 by the second law of thermodynamics.  To show equivalence of all three conditions, we proceed in the following way:
\begin{enumerate}
\item We show that condition 2 is implied by condition 1.
\item We show that condition 1 is implied by condition 2.
\item We show by construction that condition 2 implies condition 3.
\end{enumerate}

Given that $\sX$ is countable, we assume that $\sX \subseteq \N$. This is done without loss of generality: if elements of $\sX$ are not natural numbers, one can put a total order on $\sX$ using the natural numbers.  

We now prove that condition 1 implies condition 2. 
First, define the function $F$ to refer the expression in \cref{eq:prop1ep},
\begin{align}
F(\pinit) = \sum_{x} \px \big[\G(x) -\eplog \pfx + \eplog \px \big].
\label{eq:appFdef}
\end{align}
%
Let $\SetN(y)$ indicate the first $n$ elements of $f^{-1}(y)$, and  %
define the initial distribution
\[
\pinit^{(n)}(x) = \begin{cases}
e^{-\G(x)} / Z_n(y) & \text{if $x \in \SetN(y)$}\\
0 & \text{otherwise}
\end{cases} ,
\]
where $Z_n(y) =  \sum_{x \in \SetN(y)} e^{-\G(x)}$. Note that $\supp \pinit^{(n)} \subseteq \dom f$. 
Plugging into \cref{eq:appFdef} and simplifying gives  
\[
F(\pinit^{(n)}) = -\logZ_n(y) \ge 0  \,,
\]
or equivalently $Z_n(y) \le 1$. 
Since this holds for all $n$, 
\begin{align}
Z(y) &= \sum_{\mathclap{x : f(x) = y}}\; e^{-\G(x)}= \lim_{n \rightarrow \infty} Z_n(y)  \le 1 \,.
\label{eq:zresult}
\end{align}

We now prove that condition 1 is implied by condition 2. %
Define $\priorXgY(x\vert f(x))$ as in \cref{eq:mudef}, while taking $\heatToEntropy{\Q}=\G$. Then, use the results in \cref{app:EP} to rewrite $F$ as
\begin{align*}
F(\pinit) &= D(\pXgY \Vert \priorXgY) - \langle \logZ(f(X)) \rangle_{\pinit} \\
& \ge D(\pXgY \Vert \priorXgY)  \ge 0 .
\end{align*}
The first inequality follows from the assumption that $Z(y) = \sum_{x : f(x) = y} e^{-\G(x)} \le 1$ for all $y\in\img f$, and the second inequality follows from the non-negativity of conditional KL divergence~\cite{cover_elements_2012}.

\newcommand\sZ{{\mathcal{W}}}
\newcommand\yix{{\gamma}}

The rest of this proof shows by construction that condition 3 follows from condition 2. 
For simplicity, assume that the physical process has access to a set of ``auxiliary'' states, one for each $y \in \img f$. We use $x_y$ to indicate the auxiliary state corresponding to each $y$, and assume that $x_y \not \in \dom f$. For notational convenience, let $\sZ := \dom f \cup \{ x_y : y \in \img f \}$. Then, define the following function $\hat{f} : \sZ \to \sX$,
\begin{align*}
\text{For any $x\in \dom f$,}\quad& \hat{f}(x) := f(x) \,,\\
\text{For any $y\in \img f$,}\quad& \hat{f}(x_y) := y\,.
\end{align*}
In words, any $x$ in the domain of $f$ is mapped by $\hat{f}$ to $f(x)$, while any auxiliary state $x_y$ is mapped by $\hat{f}$ to $y$.

Now, 
define the following Hamiltonian $H:\sZ \to \mathbb{R} \cup \{\infty\}$,
\begin{align}
&\!\!\!\forall x \in \dom f: H(x) :=  f(x) + \entropyToHeat{ \G(x) } \label{eq:appHdef0} \\
&\!\!\!\forall y \in \img f: H(x_y) :=y - \entropyToHeat{ \ln \left(1 - \sum_{\mathclap{x:f(x) = y}} e^{-\G(x)} \right)}.  \label{eq:appHdef1}
\end{align}
We use $\pi(w) = \boltz{{H(w)}}/Z $ to indicate the Boltzmann distribution for Hamiltonian $H$, where $Z=\sum_{x \in \sZ} \boltz{{H(x)}}$ is the partition function.
Note that the partition function converges when $\beta>0$, 
\begin{align}
Z &=\sum_{y \in \img f} \Bigg[ \boltz{{H(x_y)}} + \sum_{x : f(x)=y} \boltz{{H(x)}} \Bigg] \label{eq:appL2} \\
& = \sum_{y \in \img f} \boltz{{y}}  \le \sum_{i \in\N} \boltz{{i}} = 1/(e^{\heatToEntropy{1}}-1).\nonumber
\end{align}
To derive the second line, we plugged \cref{eq:appHdef0,eq:appHdef1} into \cref{eq:appL2} and simplified.

We now consider the following physical process over $t\in [0,\ft]$, applied to a system coupled to a work reservoir and a heat bath at  temperature $T$:
\begin{enumerate}
	\item At $t=0$, the Hamiltonian $H$ is applied to the system. 

	\item Over $t \in (0,\tau]$, the system is allowed to freely relax toward equilibrium. However, the only allowed transitions are those between pairs of states $w,w'$ that have $\hat{f}(w)=\hat{f}(w')$. We assume that by $t=\tau$, the system has reached a stationary distribution.

	\item Over $t \in (\tau,\ft]$, the system undergoes a quasistatic physical process that implements the map $\hat{f}$ from initial to final states, and does so in a thermodynamically reversible way for initial distribution $\pi$.  There are numerous known ways of constructing such a process~\cite{turgut_relations_2009,maroney2009generalizing,wolpert_spacetime_2019}.
\end{enumerate}
Note that the above procedure assumes a separation of timescales (i.e., the relaxation time of the system is infinitely faster than $\tau$ and $\ft-\tau$). Step (3) also assumes an idealized heat bath (infinite heat capacity, weak coupling, infinitely fast relaxation time~\cite{deffner2013information}).

The above procedure will map any $x \in \dom f$ to final state $f(x)$.   Let $\Q$ indicate the heat function of this process. We will show that $\heatToEntropy{\Q(x)}=\G(x)$ for any $x \in \dom f$. 
First, let $\delta_x$ indicate an initial distribution which is a delta function over some state $x$. Note that
\begin{align}
\EP(\delta_x) &= S(\delta_{f(x)}) - S(\delta_x) + \heatToEntropy{\langle Q\rangle_{\delta_x}} =  \heatToEntropy{Q(x)} ,
\label{app:ep000}
\end{align}
where we've used the fact that $S(\delta_x)=S(\delta_{f(x)})=0$. 
We then analyze $\EP(\delta_x)$. 
Step (1) and step (3) in the above  construction incur no EP.  For step (2), EP incurred during free relaxation from $t=0$ to $t=\tau$ is given by
\begin{align}
\EP(\delta_x) = D(\delta_x \Vert \pi) - D(p^\tau_x \Vert \pi) \,, \label{eq:appdropkl}
\end{align}
where $p^\tau_x$ is the state distribution at time $\tau$, given that the system started in distribution $\delta_x$ at $t=0$.  By construction, $p^\tau_x$ will be equal to the equilibrium distribution restricted to a subset of states,
\begin{align*}
p^\tau_x(w) = \frac{ \delta(\hat{f}(w),\hat{f}(x))\pi(w) } {\sum_{w'} \delta(\hat{f}(w'),\hat{f}(x)) \pi(w') } .
\end{align*}
It can be verified, using the definition of $\delta_x$ and $\pi$, that
\begin{align*}
D(\delta_x \Vert \pi) = \heatToEntropy{f(x)} + G(x) + \ln Z.
\end{align*}
Similarly, it can be verified using the definition of $p^\tau_x$ that
\begin{align*}
D(p^\tau_x \Vert \pi) &= \heatToEntropy{f(x)} + \ln Z.
\end{align*}
Plugging these two KL divergences into \cref{eq:appdropkl} gives 
\begin{align}
\EP(\delta_x) = G(x).
\label{eq:ap01}
\end{align}
Combining with  \cref{app:ep000} gives $\heatToEntropy{\Q(x)} = \EP(\delta_x) = \G(x)$.  
\end{proof}

It can be verified that the physical process constructed in the proof of \cref{prop:second} is thermodynamically reversible if it is started with the initial equilibrium distribution $\pi_X$, so that the free relaxation in step 2 incurs no EP. Generally, this equilibrium distribution will have support on the auxiliary states, which are outside of $\dom f$. However, consider the case when \cref{eq:prop1ineq} is an equality for all $y \in \img f$. Then, the definition in \cref{eq:appHdef1} gives $H(x_y)=\infty$ and $\pi_X(x_y) = 0$ for all $y\in\img f$.  In this case, the input distribution $p_X = \pi_X$ obeys $\supp p_X \subseteq \dom f$ and achieves zero EP. Moreover, using the decomposition in \cref{app:EP}, it can be verified that if \cref{eq:prop1ineq} is an equality for all $y\in \img f$, then any input distribution that obeys $\pXgY = \pi_{X|f(X)}$, as defined in \cref{eq:appdefpxy}, 
also achieves zero EP.


%
%
%
%
%
%
%
%
%
%
%
%
%
%
%
%
%
%
%
%
%
%
%
%
%
%

%
%
%
%
%
%
%
%
%
%
%
%
%
%

\section{$O(1)$ heat for coin-flipping realization is uncomputable}
\label{app:nootherconst}

Let $\funcU$ indicate the partial function computed by some UTM $\U$. 
Imagine there is some computable function $f$ such that for any $y$, $f(y)$ returns an input for $\funcU$ that outputs $y$ and generates bounded heat under the coin-flipping realization (i.e., $\funcU(f(y))=y$ and $\Qcoin(f(y))=O(1)$). Then,  by \cref{eq:coinwork}, it must be that $\ell(f(y)) = K(y)+ O(1)$. Since $\ell(\cdot)$ is a computable function, this would in turn imply that there is a computable function $g(y) =  K(y)+ O(1)$. However, such a function cannot exist, as shown in the following proposition.

\begin{proposition}
\label{prop:nootherconst}
There is no computable partial function $g : \BB \to \N$ such that for all $y$,
\begin{align}
g(y) = K(y) + O(1) \,.
\end{align}
\end{proposition}
\begin{proof}
We say that $p_Y(y)$ is a \emph{semimeasure} if $p_Y(y) \ge 0$ for all $y$ and $\sum_y p_Y(y) \le 1$ (i.e., it is a non-normalized probability distribution). We say that a semimeasure $p_Y(y)$ \emph{(multiplicatively) dominates} another semimeasure $q_Y(y)$ if there is some constant $c >0$ such that $p_Y(y) \ge c \cdot q_Y(y)$ for all $y$.

Assume that a computable $g(y) = K(y) + O(1)$ exists. Then, $q_Y(y):=2^{-g(y)}$ would be a computable semimeasure that dominates $p_Y(y) := 2^{-K(y)}$. It is known that $p_Y(y)$ dominates every computable semimeasure~\cite[Thm.~4.3.3 and Cor.~4.3.1]{livi08}. Since domination is transitive, if $g(y)$ were computable then $q_Y(y)$ would be a computable semimeasure that dominates every computable semimeasure.  However, such a semimeasure cannot exist by Lemma~4.3.1 in \cite{livi08}.
\end{proof}

\section{Proof of \cref{eq:qoptineq}}
\label{app:dom}

Let $f$ indicate any computable partial function. 
In this appendix, we show that the dominating realization of $f$, with heat function
\begin{align}
\Qopt(x) = \ktlntwo \cdot K(x\vert f(x)),
\label{eq:appqopt}
\end{align}
is better than any other realization of $f$ with an upper-semicomputable heat function $\Q$, up to an additive constant.  

We first prove the following two useful results.

\begin{lemma}
\label{lem:applem2}
For any partial function $f : \BB \to \BB$, 
$$\sum_{\mathclap{x : f(x)=y}}e^{-\ln 2 \cdot K(x\vert y)} \le 1 \qquad \forall y\in\img f.$$
\end{lemma}
\begin{proof}
For all $y \in \img f$, we have the following: 
\begin{align*}
\sum_{\mathclap{x : f(x)=y}}e^{-\ln 2 \cdot K(x\vert y)} = \sum_{\mathclap{x : f(x)=y}} 2^{-K(x\vert y)} \le \sum_{\mathclap{x \in \BB}} 2^{-K(x\vert y)} .
\end{align*}
In addition, we have the bound
\begin{align}
  \sum_{\mathclap{x \in \BB}} 2^{-K(x\vert y)} \le 1,
  \label{eq:appkraft}
\end{align}
which comes from Kraft's inequality and the fact that, for any given $y$, the set $\{K(x\vert y) : x\in \BB\}$ specifies the lengths of a prefix-free code~\cite[p.~252 and p.~287]{livi08}.  Combining gives the desired result.
\end{proof}

\newcommand{\pbQ}{\Q}
\begin{proposition}
\label{prop:best}
Let $f : \BB \to \BB$ be a  computable partial function, $\pbQ : \BB \to \R$ a upper-semicomputable partial function with $\dom \pbQ \supseteq \dom f$. If for all $y \in \img f$,
\begin{align}
\sum_{x : f(x)=y} e^{-\pbQ(x)} \le 1 ,
\label{eq:ineqapp9}
\end{align}
then for all $x\in \dom f$,
\begin{align}
\pbQ(x) \ge \ln 2 \cdot [ K(x \vert f(x)) - K(\pbQ,f)] + O(1) \,,
\label{eq:bestres0}
\end{align}
where $O(1)$ is a constant independent of $x$ and $\pbQ$.
\end{proposition}
\begin{proof}
\def\stepN{n}
Let $\T$ indicate the TM that computes $f$, and let $a(x,\stepN)$ be a computable partial function 
which upper-semicomputes $\pbQ(x) / \ln 2$. 
Then, define the following TM $B$:  
given inputs $x\in \BB$, $y\in \BB$, and $\stepN\in \mathbb{N}$, the TM $B$ runs $\T$ for $\stepN$ steps on input $x$. If $\T$ halts within that time on output $y$, then $B$ outputs $2^{-a(x,n)}$.  Otherwise, $B$ outputs 0 and halts. 

Then, for any $x \in \dom f$, define
\begin{align}
s(x|y) & := \lim_{\stepN \rightarrow \infty} \func{B}(\langle x,y\rangle, \stepN)\nonumber\\
& = \delta(f(x),y) 2^{-\pbQ(x) / \ln 2} 
\nonumber \\
& = \delta(f(x),y) e^{-\pbQ(x)} .
\label{eq:appA4}
\end{align}
It is easy to check that $\func{B}(\langle x,y\rangle, \stepN)$ is non-decreasing in $\stepN$, so $s(x|y)$ is lower-semicomputable   
(i.e.,  $\func{B}(\langle x,y\rangle ,n) \le \func{B}(\langle x,y\rangle,n+1)$ and $\lim_{n\to \infty} \func{B}(\langle x,y\rangle ,n) = s(x\vert y)$).   
Moreover, 
if one had a program that computed both $f$ and $\pbQ$, then one could lower-semicompute $s$. This means that 
\begin{align}
K(s) \le K(\pbQ,f) + O(1), 
\label{eq:appA5}
\end{align}
where $K(\pbQ,f)$ is the Kolmogorov complexity of jointly computing the functions $f$ and $\pbQ$. 

By assumption in \cref{eq:ineqapp9}, for any $y \in \img f$,
\begin{align}
\sum_{\mathclap{x \in \dom f}} s(x\vert y) = \sum_{\mathclap{x : f(x) = y}} 2^{-\pbQ(x)/\ln 2} = \sum_{\mathclap{x : f(x) = y}} e^{-\pbQ(x)}  \le 1 .
\label{eq:app.A.5}
\end{align}
This means that $s(x|y)$ is a so-called \termdef{conditional semimeasure} of $x$ given $y$ (i.e., a non-normalized conditional probability measure). For any lower-semicomputable conditional semimeasure $s$, an existing result in AIT~\cite[Cor.~2]{vitanyi2013conditional} states 
\begin{align*}
K(x\vert y) & \le -\log_2 s(x| y) + K(s) + O(1).\label{eq:appK0}
\end{align*}
Taking $y= f(x)$ and plugging in \cref{eq:appA4,eq:appA5} gives
\begin{align}
K(x\vert f(x))& \le \pbQ(x)/\ln 2 + K(\pbQ,f) + O(1).
\end{align}
\cref{eq:bestres0} follows by rearranging.
%
%
\end{proof}

Given that \cref{eq:ineqapp9} holds, by \cref{prop:second} there must be a realization of $f$ with heat function $Q(x) = kT G(x)$. By \cref{lem:applem2}, we can take $G(x) = \bitsToEntropy{ K(x \vert f(x)) }$. Thus, 
there must exists a realization of $f$ with heat function $\Qopt$, as defined in \cref{eq:appqopt}.

Combining $G(x) = \Q(x)/kT$ with \cref{eq:bestres0}, and multiplying both sides by $kT$, gives the following inequality, 
$$
\Q(x) \ge \Qopt(x) - \bitsToHeat{ K(\heatToEntropy{\Q},f) } + O(1) \,.
$$
We can derive a slightly weaker, but more interpretable, lower bound by using $K(\heatToEntropy{\Q},f) \le K(\heatToEntropy{\Q}) + K(f) + O(1)$, which follows from the subadditivity of Kolmogorov complexity~\cite[p.202]{livi08}.  This allows to rewrite the above as
$$
\Q(x) \ge \Qopt(x) - \bitsToHeatParen{ K(\heatToEntropy{\Q}) + K(f)} + O(1) ,
$$
which appears in the main text as \cref{eq:qoptineq}, with $f=\funcT$.

\section{Infinite expected heat}
\label{app:infiniteheat}

Let $\funcU$ be the partial function computed by some UTM $U$. 
In the following results, we will make use of the following decomposition of the drop of entropy, which holds for any initial distribution $\pinit$:
\begin{align}
S(\pinit) - S(\pfin) & = %
 \sum_{\mathclap{y \in \img \funcU}} \pfin(y) S(p_{X\vert \funcU(x)=y}). \label{eq:appdecomp}
\end{align}
Note that (discrete) Shannon entropy is non-negative, so $S(p_{X\vert \funcU(x)=y})\ge 0$ for all $y$. 
For simplicity, and without loss of generality, in this section we will write Shannon entropies in units of bits.

We will make use of the following lemmas.

\begin{lemma}
\label{lem:infapp}
For any $y\in \BB$, 
$$\sum_{\mathclap{x : \funcU(x) = y}} 2^{-\ell(x)} \ell(x) = \infty.$$
\end{lemma}
\begin{proof}
To derive this result, we make use of a simple prefix-free code for natural numbers $i\in \mathbb{N}$:
\begin{align}
g(i) =\underbrace{111...111}_{\text{\ensuremath{\lceil\log_{2}i\rceil\;1\text{s}}}}0\underbrace{1110...0110}_{\substack{\text{Encoding of \ensuremath{i}}\\
\text{with \ensuremath{\lceil\log_{2}i\rceil} bits}
}} .\label{eq:numbercode}
\end{align}
(See also \cite[Section 1.11]{livi08}.)
It is straightforward to check that this prefix-free code achieves a code length 
\begin{align}
\ell(g(i)) = 2\lceil \log_2 i\rceil + 1 .
\label{eq:ncodelen}
\end{align} 
In addition, we will also use programs of the form $z_y + g(i) + x$ such that $\funcU(z_y + g(i) + x) = y$, where $z_y$ is some appropriate prefix string, $g(i)$ is defined in \cref{eq:numbercode},  $x$ is any binary string with $\ell(x)=i$, and ``$+$'' indicates concatenation. In words, the program $z_y + g(i) + x$ causes $U$ to read in a code for $y$ (corresponding to $z_y$), then a prefix-free code for any $i \in \mathbb{N}$ (corresponding to $g(i)$), then ``swallow'' $i$ bits of input (corresponding to $x$), and halt after outputting $y$. Using \cref{eq:ncodelen}, it can be checked that 
\begin{multline}
i =\ell(x) < \ell(z_y + g(i) + x) = \ell(z_y) + \ell(g(i)) + i \\\le \ell(z_y) + 2 \log_2 i + 3 + i  \,.
\label{eq:codeineqs}
\end{multline}

We now bound the sum $\sum_{x : \funcU(x) = y} 2^{-\ell(x)} \ell(x)$. Since all terms in this sum are positive, we can lower bound it by focusing only on the subset of programs of the form $z_y + g(i) + x$:
\begin{align*}
\sum_{\mathclap{x : \funcU(x) = y} } 2^{-\ell(x)} \ell(x) &\ge \sum_{\mathclap{i\in \mathbb{N}, x: \ell(x) = i}} 2^{-\ell( z_y + g(i) + x )} \ell(z_y + g(i) + x) \\
& \stackrel{(a)}{\ge}  \sum_{\mathclap{i\in \mathbb{N}, x: \ell(x) = i}} 2^{-\ell( z_y) - 2 \log_2 i - 3- i } i \\
& = 2^{-\ell( z_y) - 3} \sum_{\mathclap{i\in \mathbb{N}, x: \ell(x) = i}} 2^{- 2 \log_2 i- i } i \\
& \stackrel{(b)}{=} 2^{-\ell( z_y) - 3} \sum_{\mathclap{i\in \mathbb{N}}} 2^i 2^{- i - 2 \log_2 i} i \\
& = 2^{-\ell( z_y) - 3} \sum_{\mathclap{i\in \mathbb{N}}} i/i^2 \\
& = 2^{-\ell( z_y) - 3} \sum_{\mathclap{i\in \mathbb{N}}} 1/i = \infty\,.
\end{align*}
In $(a)$, we use the lower and upper bounds on $\ell(z_y + g(i) + x)$ from \cref{eq:codeineqs}, and in $(b)$ we use  that there are $2^i$ different bit strings $x$ that obey $\ell(x) = i$.   The rest of the steps follow from rearranging and simplifying.
\end{proof}

%
\begin{lemma}
\label{lem:kbound}
For any computable partial function $f:\BB \to \BB$ and $x \in \dom f$, 
$$K(x) \le \ell(x) + O(1).$$
\end{lemma}
\begin{proof}
Let $\T$ be a TM which computes $f$, and note that $\dom \funcT$ is a prefix-free set. 
Consider the Kolmogorov complexity $K_\U(x)$, which is defined in terms of a UTM $\U$ which operates in the following way: $\U$ takes inputs of the form $b + x$, where $b\in\{0,1\}$, $x\in \BB$ and ``$+$'' indicates string concatenation.  If $b=0$, then $U$ emulates some prefix UTM on input $x$ and outputs the result.  If $b=1$, then $U$ emulates $\T$ on input $x$ while swallowing the output; if and when $\T$ halts on input $x$, $U$ outputs a copy of the input $x$ and halts.  It is clear that $U$ is universal, due to its behavior when $b=0$, and that it is prefix-free.  It is also clear that $U$ has a program of length $\ell(x)+1$ that can be used to output any $x\in\dom \funcT$, due to its behavior when $b=1$. Thus, $K_U(x) \le \ell(x)+1$. The result follows by recalling the Invariance Theorem, $K(x) =K_U(x) + O(1)$, where $K(\cdot)$ is the Kolmogorov complexity defined for some arbitrary reference UTM.
\end{proof}

\subsubsection{Coin-flipping distribution}
In this section, we consider the coin-flipping input distribution, $\pcoin$, as defined in \cref{eq:coin0}. We show that the drop in entropy for this input distribution is infinite,
\begin{align}
S(\pcoin) - S(\pfincoin) = \infty \,.
\label{eq:appinf0}
\end{align}
Thus, by the second law of thermodynamics, \cref{eq:ep1_alternative}, any realization which carries out $U$ on $\pcoin$ must generate an infinite amount of heat.

To derive \cref{eq:appinf0}, first use \cref{eq:appdecomp} to write
\begin{align}
S(\pcoin) - S(\pfincoin)  = \sum_{\mathclap{y \in \img U}} \pfincoin(y) S(p_{X\vert \funcU(x)=y}^\mathrm{coin}).
\end{align}

We now show that $ S(p_{X\vert \funcU(x)=y}^\mathrm{coin}) = \infty$ for any $y\in\img \funcU$. First, write
\begin{align}
&S(p_{X\vert \funcU(x)=y}^\mathrm{coin}) = - \sum_{\mathclap{x : \funcU(x) = y}} p^\mathrm{coin}_{X\vert Y}(x\vert y) \log_2 p^\mathrm{coin}_{X\vert Y}(x\vert y) \nonumber \\
&\quad = \log_2 \uY(y) - \frac{1}{\uY(y)} \sum_{{x : \funcU(x) = y} } 2^{-\ell(x)} \log_2 2^{-\ell(x)}  \nonumber \\
& \quad = \log_2 \uY(y) + \frac{1}{\uY(y)} \sum_{{x : \funcU(x) = y} } 2^{-\ell(x)} \ell(x) \label{eq:appllline3}
\end{align}
where we use that $p^\mathrm{coin}_{X\vert Y}(x\vert y)=2^{-\ell(x)}/\uY(y)$ when $\funcU(x)=y$ (similarly to the derivation in \cref{sec:EF_coin_flipping}).  
Note that the multiplicative constant $1/\uY(y)$ is strictly positive, and the additive constant $\log_2 \uY(y)$ is finite.  Then, \cref{eq:appllline3} is infinite by \cref{lem:infapp}.

\subsubsection{EP optimal distribution for the dominating realization}

\def\aY{w_Y}
\def\C{C}

Consider any initial distribution of the form 
\begin{align}
\pinit(x) = \frac{\aY(\funcU(x))}{\C(\funcU(x))} 2^{-K(x\vert \funcU(x))} \,,
\label{eq:appdist0}
\end{align}
where $\C(y) := \sum_{x:\funcU(x)=y} 2^{-K(x\vert y)}$ is a normalization constant, and $\aY$ is any probability distribution over  $\img \funcU$. It can be verified, using results discussed in section \cref{app:EP}, that  any input distribution of the form \cref{eq:appdist0} achieves 0 mismatch cost for the dominating realization.  Thus, this distribution achieves minimal EP for the dominating realization.

In this section, we show that any input distribution of the form \cref{eq:appdist0}  also incurs an infinite drop in entropy
\begin{align}
S(\pinit) - S(\pfin) = \infty \,.
\label{eq:appinf0b}
\end{align}
Thus, by the second law of thermodynamics, \cref{eq:ep1_alternative}, any realization which carries out $\U$ on such an input distribution $\pinit$ must generate an infinite amount of heat.

Our derivation proceeds in a similar manner as that used above to show that the drop in entropy for $\pcoin$ was infinite.  First, use \cref{eq:appdecomp} to write
\begin{align}
S(\pinit) - S(\pfin)  = \sum_{y \in \img \funcU} \pfin(y) S(p_{X\vert \funcU(x)=y})
\end{align}

We derive \cref{eq:appinf0b} by showing that $ S(p_{X\vert \funcU(x)=y}) = \infty$ for any $y\in\supp \aY$.   First, write
\begin{align}
&S(p_{X\vert \funcU(x)=y}) = - \sum_{\mathclap{x : \funcU(x) = y}} p_{X\vert Y}(x\vert y) \log_2 p_{X\vert Y}(x\vert y) \nonumber \\ 
& \quad = \log_2 \C(y) + \frac{1}{\C(y)} \sum_{{x : \funcU(x) = y} } 2^{-K(x \vert y)} K(x\vert y) \label{eq:app15}
\end{align}
where we use that $p_{X\vert Y}(x\vert y)=2^{-K(x\vert y)}/\C(y)$ when $\funcU(x)=y$ and $\aY(y)>0$. 
To show that \cref{eq:app15} is infinite, we note that $\C(y) > 0$, and then focus on the inner sum 
\begin{align}
\sum_{\mathclap{x : \funcU(x) = y} }  2^{-K(x\vert y)} K(x \vert y)
\label{eq:appsum9b}
\end{align}
Note that any $x$ such that $\funcU(x)=y$ must obey $x\in \dom \funcU$. This means that 
\begin{align*}
K(x\vert y) \le K(x) + O(1) \le \ell(x) + O(1),
\end{align*}
where the first inequality comes from subadditivity of Kolmogorov complexity~\cite{livi08}, while the second comes from \cref{lem:kbound}. 
We will use $\kappa \ge 0$ to indicate some finite constant that makes the rightmost inequality hold. 

Now, note that $2^{-a} a$ is non-increasing in $a\in \N$ for all $a \ge 1$. Assume for the moment that there is no $x$ such that $\funcU(x)=y$ and $K(x\vert y) = 0 $. Then,
\begin{align*}
2^{-K(x\vert y)} K(x \vert y) &\ge  2^{-\ell(x)-\kappa }(\ell(x) + \kappa ) \ge 2^{-\kappa} 2^{-\ell(x) }\ell(x)
\end{align*}
for all $x$ such that  $\funcU(x)=y$. This gives the following lower bound for \cref{eq:appsum9b}:
\begin{align*}
\sum_{\mathclap{x : \funcU(x) = y} }  2^{-K(x\vert y)} K(x \vert y) \ge 2^{-\kappa} \sum_{\mathclap{x : \funcU(x) = y} } 2^{-\ell(x) }\ell(x) = \infty ,
\end{align*}
where the last equality uses \cref{lem:infapp}.  Now imagine that there is an $x$ such that $\funcU(x)=y$ and $K(x\vert y)=0$ (for any given $y$, there can be at most one such $x$). In that case, the above lower bound should be decreased by $2^{-\kappa} 2^{-\ell(x) }\ell(x)$, which is a finite constant, so \cref{eq:appsum9b} is still infinite.


\section{Strictly positive EP for the dominating distribution}
\label{app:posEP}

Consider any computable partial function $f$, and recall the decomposition of EP developed in \cref{app:EP}, into a non-negative ``mismatch cost'' (conditional KL) term and a non-negative ``residual EP'' term, \cref{eq:decomp2}.  The residual EP term is an expected over non-negative values $-\logZ(y)$ for $y \in \img f$.  

Using \cref{eq:zdef}, we write this residual term for the dominating realization as 
\begin{align*}
- \logZ(y) &= -\eplog \sum_{\mathclap{x : f(x)=y}} e^{-\heatToEntropy{\Qopt(x)}} = -\eplog \sum_{\mathclap{x : f(x)=y}} 2^{-K(x\vert y)},
\end{align*}
where we substituted in the definition of $\Qopt$ from \cref{eq:qoptdef}. Assume that the conditional Kolmogorov complexity is defined relative to some reference UTM $U$, $K(x\vert y)=K_U(x\vert y)$.  Then, consider the inner sum,
\begin{align*}
\sum_{{x : \funcU(x)=y}} 2^{-K_U(x\vert y)} &\le \sum_{{x \in \BB}} 2^{-K_U(x\vert y)} \\
& < \sum_{{(z,y) \in \dom \funcU}} 2^{-\ell(z)}  \le 1.
\end{align*}
The strict inequality comes from the fact that not all programs  $(z,y)\in\dom \funcU$ are the shortest program for some output string $x \in \BB$.  The last inequality comes from the Kraft inequality.  


This shows that for the dominating realization of a computable function $f$, $-\logZ(y)>0$ for all  $y \in \img f$. Thus, the residual EP term in \cref{eq:decomp2} is strictly positive for any input distribution.

\section{Derivation of \cref{eq:logbetter}}
\label{app:domcoin}

%
%
%
%
%
%
%
For a coin-flipping realization of some UTM $U$, \cref{eq:coinwork} 
states that the heat generated on input $x$ is given by 
\begin{align*}
\Qcoin(x) &= \bitsToHeatParen{ \ell(x) - K(\funcU(x))} + O(1)\\
&\ge \bitsToHeatParen{ K(x) - K(\funcU(x))} + O(1),
\end{align*}
where the second line uses
\cref{lem:kbound}. 
We now use the following inequality~\cite[Sec.~3.9.2]{livi08}:
\begin{align*}
&K(\funcU(x)) \\
& \le K(x,\funcU(x)) - K(x\vert \funcU(x))+ O(\log K(\funcU(x)))\\
& = K(x) - K(x\vert \funcU(x)) + O(\log K(\funcU(x))),
\end{align*}
where in the last line we've used that $K(x,\funcU(x))=K(x) + O(1)$ (since the value of $\funcU(x)$ is by definition computable from $x$). 
Combining the above results with the definition of $\Qopt$ gives the desired result,
\begin{align}
\Qcoin(x) \ge \Qopt(x) - O(\log K(\funcU(x))). 
\end{align}

\clearpage

\end{document}